\documentclass{article}
\usepackage{amsfonts,amssymb,mathrsfs}

\usepackage{graphicx,color}
\usepackage{amsmath,epstopdf}

\newtheorem{theorem}{Theorem}

\newtheorem{definition}[theorem]{Definition}

\newtheorem{lemma}[theorem]{Lemma}
\newtheorem{observation}[theorem]{Observation}
\newtheorem{notation}[theorem]{Notation}

\newtheorem{proposition}[theorem]{Proposition}

\newenvironment{proof}[1][Proof]{\textbf{#1.} }{\ \rule{0.5em}{0.5em}}

\definecolor{green}{rgb}{0.00,0.50,0.00}

\begin{document}

\title{Non-Markovian Quantum Feedback Networks II: Controlled Flows}
\author{John E.~Gough \footnote{jug@ber.ac.uk}\\
Aberystwyth University, SY23 3BZ, Wales, United Kingdom}
\date{}
\maketitle

\begin{abstract}
The concept of a controlled flow of a dynamical system, especially when the controlling process feeds information back about the system, is of central importance in control engineering. In this paper we build on the ideas presented by Bouten and van Handel (L. Bouten, R. van Handel, \lq\lq On the separation principle of quantum control", in \emph{Quantum Stochastics and Information: Statistics, Filtering and Control}, World Scientific, 2008) and develop a general theory of quantum feedback. We elucidate the relationship between the controlling processes $Z$ and the measured process $Y$, and to this end make a distinction between what we call the input picture and the output picture. We should that the input-output relations for the noise fields have additional terms not present in the standard theory, but that the relationship between the control processes and measured processes themselves are internally consistent - we do this for the two main cases of quadrature measurement and photon-counting measurement. The theory is general enough to include a modulating filter which processes the measurement readout $Y$ before returning to the system. This opens up the prospect of applying very general
engineering feedback control techniques to open quantum systems in a systematic manner, and we consider a number of specific modulating filter problems. Finally, we give a brief argument as to why most of the rules for making
instantaneous feedback connections (J. Gough, M.R. James, \lq\lq Quantum Feedback Networks: Hamiltonian Formulation", Commun. Math. Phys. \textbf{287}, 1109, 2009) ought to apply for controlled dynamical networks as well.
\end{abstract}

\section{Introduction}
Markovianity is one of the most frequently discussed and misunderstood ideas in physics. The essential problem is the concept of 
the state of a system, and this goes back to Hamilton who took Newton's second order differential equations of motion and converted them into first order equations for the mechanical state (the positions and momenta). The subsequent development of mechanics by
Hamilton, Jacobi, Liouville and Poisson showed that the collection of states (phase space) was more than just a blank canvas, but had its own canonical structure. At its heart, Hamiltonian mechanics is as far reaching as it is because it tells us how to define the
appropriate state and how to propagate it forward in time. Markov effectively extended this to stochastic systems: one again has a state space for the system, and a transition mechanism to tell us the probability to get from one state to another in a given time. 
The propagation of state for stochastic systems can frequently be presented as an explicit first order (stochastic)
differential equations, or dilation. 
Here one makes idealized assumptions on the noise (for instance, that it is Wiener) and builds up the model
from a mathematical basis trying to capture the real world problem as accurately as possible. The advantage has been that the
models are then tractable and allow for the formulation and solution of sophisticated control and optimization problems.

The programme of extending these ideas into the quantum world is an old one, with Dirac being one of the first to grasp the
canonical structure behind quantum mechanics and its structural relation to Hamiltonian mechanics. The systematic 
extension to open quantum systems came about in 1970's and 80's, with the key contribution being the development of quantum 
stochastic calculus by Hudson and Parthasarathy \cite{HP84}, see also \cite{Par92}, which enabled the explicit construction of dilations of completely positive quantum dynamical semi-groups. 

This is the second in as series of papers looking a generalizations of quantum feedback networks \cite{GouJam09a} where the Markov
property is either pushed to its limit, our replaced by physical models. In the first paper of the series \cite{NM_QFN_I},
we looked at models arising from quantizing transmission lines in a manner that was fundamentally non-Markov. We now turn
to models which use the Hudson-Parthasarathy theory, but now allow the coefficients of the quantum stochastic differential equations 
to be adapted quantum stochastic processes. More specifically, we want these coefficients to depend in a causal manner on
a commutative quantum stochastic process which we call the control process. In this way we obtain a controlled quantum flow when
we look at the stochastic Heisenberg dynamics. What we would like to do is to perform some continuous time measurement on the 
output noise from the system and take the control process to be some filtered process obtained from the measurement readout
process. This is related to the theory of quantum non-demolition measurements due to Belavkin \cite{Belavkin1}, however, our
interest is with the internal consistency of the feedback mechanism and do not need explicitly the quantum filtering (trajectories)
theory. In principle, the filter modulating the measurement readout process, before feeding back into the system, may be described
classically and one may worry about mixing classical and quantum dynamical systems. The treatment we give however circumvents this
issue.

It is not uncommon for theoretical physicists to refer to a master equation with time-dependent coupling operators as non-Markovian.
As such, our model ought to be considered as radically non-Markovian as we not only allow the coupling operators to be time-dependent,
we allow them to depend of the past history of the measurement readout. However, the theory we present is firmly set in the 
framework of the Hudson-Parthasarathy theory and makes full use of calculus used to build Markov models and retains all the
features of being able to propagate the state forward in time, albeit with a much more complicated feedback-modulated dynamical flow.
The preference  of the author would be to describe the models here as pushing the quantum Markov description about as far as possible.

\begin{figure}[h]
	\centering
		\includegraphics[width=0.50\textwidth]{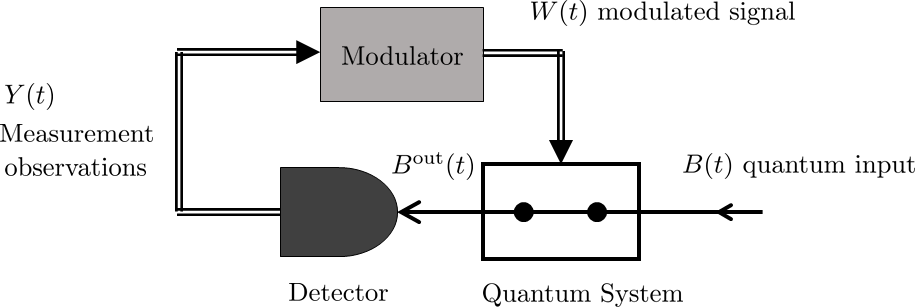}
	\caption{The output field of a quantum system is measured and the measurement observations $Y$ processed by a modulator.
	The $SLH$-coefficients may then be considered as functions of the modulated process $W$.}
	\label{fig:NM_QFN_II_Fb}
\end{figure}

The use of measurement feedback in quantum optics was first proposed by Wiseman \cite{Wiseman}. This is an early form of what has
now become classified as environment engineering. Subsequently, models were developed for realistic photo-detectors \cite{WWM}
which took account of the modification of the measurement output that would likely occur in the detection - so that the real detector was effectively an ideal detector cascaded with a filter. In practice, one would want to design a filter to process the measurement
readout before feeding back to the system. Here one needs the concept of a controlled dynamics, and in the context of open quantum 
systems this was first formalized by Bouten and van Handel \cite{BvH,BvH_ref}, and also \cite{BouvanHJam07}. This paper builds on their framework.

This situation is depicted in Figure \ref{fig:NM_QFN_II_Fb}. We have a quantum system $%
G\sim \left( S,L,H\right) $ with input process $B\left( t\right) $. The
output field, $B^{\text{out}}\left( t\right) $, is passed to a detector
which measures, say, the quadrature $Y\left( t\right) =B^{\text{out}}\left(
t\right) +B^{\text{out}\ast }\left( t\right) $. For an excellent recent account of what is now referred to as the $SLH$-theory
see the review by Combes, Kerckhoff and Sarovar \cite{CKS}. The measured output is
essentially classical and may be treated as a classical stochastic process -
it may then be passed through a modulator which, for instance, may smooth
the measured output $Y(t)$. The modulated process $W\left( t\right) $ will
be a process that depends causally on $Y\left( t\right) $ - for stochastic
processes, $W$ will be adapted to the filtration generated by $Y$. The $SLH$%
-coefficients are then made to depend on $W\left( t\right) $.

The treatment here relies on making appropriately engineered interconnections between the system, the detector and the modulator.
In this respect, it is closer to coherent quantum feedback models \cite{YK}-\cite{Yamamoto14} than measurement feedback 
\cite{WM_book}, however we show the consistency of these approaches.

In section \ref{sec:ConDyn}, we give an abstract account of controlled quantum flows where we assume that the coefficients of
the quantum stochastic differential equation - the $SLH$-coefficients - are allowed to be adapted processes depending on some control
process $Z$, itself an adapted commutative quantum stochastic process on the noise space. We pay particular attention to the form of
the input-output relations for the noise field as this is quite different from the usual (autonomous Markov!) case where the 
coefficients are fixed operators on the system space. Our main contribution here is to elucidate the difference between the 
\textit{input picture} and the \textit{output picture}: in particular, the measured output $Y$ should be control process $Z$
rotated from input to output picture using the unitary quantum stochastic process. In section \ref{sec:Quad}, we focus
on the special case of quadrature measurement, and in this case look at several explicit models where the measurement readout $Y$
is modulated before feedback to the system.

\subsection{Quantum Stochastic Evolutions}

The idea of stochastic dilations originates in the observation that heat equations - (hypo)-elliptic 2nd order
\textit{partial} differential equations - may be treated by instead considering 1st order \textit{stochastic}
differential equations. In a non-commutative setting, say a C*-algebra, the role of 1st differential operators (derivations) 
is played by maps of the form $ [ \cdot ,  K ]$ for a fixed element $K$. Hudson-Parthasarathy developed
a theory of quantum stochastic evolutions which gave explicit dilations of non-commutative heat equations - GKS-Lindblad 
master equations. The setting of their theory is a Hilbert space of the form 
\begin{eqnarray}
\mathfrak{H}=\mathfrak{h}_{0}\otimes \Gamma \left( \mathfrak{K}\otimes
L^{2}[0,\infty )\right)
\end{eqnarray}
where $\mathfrak{h}_{0}$ is a fixed Hilbert space, called the \emph{initial
space}, and $\mathfrak{K}$ is a fixed Hilbert space called the \emph{%
internal space} or \textit{multiplicity space}. Here we want a finite number, $n$, of boson inputs labeled by the set $\mathsf{k}$, and to this end, $\mathfrak{K}=\mathbb{C}^{\mathsf{k}}$ and $\mathsf{k}=\left\{ 1,\cdots ,n\right\} $.
The second quantization functor is here denoted as $\Gamma (\cdot )$ and the theory makes extensive use of the tensor product
decomposition $ \Gamma (\mathfrak{k}_1 \oplus \mathfrak{k}_2 ) \cong \Gamma (\mathfrak{k}_1   )
\otimes \Gamma (  \mathfrak{k}_2 )$, for arbitrary Hilbert spaces $\mathfrak{k}_1 $ and $ \mathfrak{k}_2 $.

Take $\left\{ e_{k}:k\in \mathsf{k}\right\} $ be the canonical orthonormal basis for $\mathfrak{K}$, then the $k$th
annihilation process is defined to be
\begin{eqnarray*}
B_{k}\left( t\right) \triangleq a\left( e_{k}\otimes 1_{\left[ 0,t\right]
}\right)
\end{eqnarray*}
where $a\left( \cdot \right) $ is the annihilation functor from $\mathfrak{K}%
\otimes L^{2}[0,\infty )$ to the Fock space $\Gamma \left( \mathfrak{K}%
\otimes L^{2}[0,\infty )\right) $ and $1_{\left[ 0,t\right] }$ is the
indicator function for the interval $\left[ 0,t\right] $. We denote its adjoint, the $k$th creation process as
$B_{k}\left( t\right) ^{\ast }$. The scattering process from the $k$th to the $j$th field is defined to be 
\begin{eqnarray*}
\Lambda _{jk}\left( t\right) \triangleq d\Gamma \left( |e_{j}\rangle \langle
e_k|\otimes \pi _{\left[ 0,t\right] }\right)
\end{eqnarray*}
where $d\Gamma \left( \cdot \right) $ is the differential second
quantization functor and $\pi _{\left[ 0,t\right] }$ is the operator of
pointwise multiplication by $1_{\left[ 0,t\right] }$ on $L^{2}[0,\infty )$.

The tensor product decomposition then implies the \textit{continuous tensor product decomposition}
\begin{eqnarray}
\mathfrak{H} \cong \mathfrak{H}_{\left[ 0,t\right] }\otimes \mathfrak{H}_{\left( t,\infty \right) },
\label{eq:ctpd}
\end{eqnarray}
for each $t>0$,
where $\mathfrak{H}_{\left[ 0,t\right] }=\mathfrak{h}_{0}\otimes \Gamma
\left( \mathfrak{K}\otimes L^{2}[0,t)\right) $ and $\mathfrak{H}_{\left(
t,\infty \right) }=\Gamma \left( \mathfrak{K}\otimes L^{2}(t,\infty )\right) 
$. We shall write $\mathfrak{A}_{t]}$ for the space of bounded operators on $%
\mathfrak{H}$ that act trivially on the future component $\mathfrak{H}%
_{\left( t,\infty \right) }$, that is
\begin{eqnarray}
\mathfrak{A}_{t]} \equiv \mathscr{B} ( \mathfrak{H}_{\left[ 0,t\right] }) ,
\label{eq:A_t]}
\end{eqnarray}

Following Hudson and Parthasarathy, we refer to a family $X_{t}=\left\{ X_{t}:t\geq 0\right\}$
of operators on $\mathfrak{H}$ as a quantum stochastic process, and we say that the process is \emph{adapted} if $X_{t}\in %
\mathfrak{A}_{t]},$ for each $t\geq 0$. They then proceed to define quantum stochastic integrals, in the sense of It\={o},
of adapted processes with respect to the creation, annihilation, scattering processes.
Taking $\left\{ x_{\alpha \beta }\left( t\right) :t\geq 0\right\} $ to be a
family of adapted quantum stochastic processes, their quantum stochastic
integral is $X_{t}=\int_{0}^{t}x_{\alpha \beta }\left( s\right) dB^{\alpha
\beta }\left( t\right) $ which is shorthand for 
\begin{eqnarray*}
\int_{0}^{t}x_{00}\left( s\right) ds+\sum_{j\in \mathsf{k}}\int_{0}^{t} x_{j0}\left(
s\right) dB_{j}\left( s\right) ^{\ast } 
+\sum_{k\in \mathsf{k}} \int_{0}^{t} x_{0k}\left(
s\right) dB_{k}\left( s\right) +\sum_{j,k\in \mathsf{k}} \int_{0}^{t}
x_{jk}\left( s\right) d\Lambda _{jk}\left( s\right) .
\end{eqnarray*}
The differentials are understood in the It\={o} sense: for each $s$ the coefficient is an operator in the past algebra
$\mathfrak{A}_{s]}$ while the increments are future pointing and act non-trivially on $\mathfrak{H}_{(s, \infty )}$.
Given a second quantum It\={o} integral $Y_{t}$, with $dY_{t}=y_{\alpha \beta
}\left( t\right) dB^{\alpha \beta }\left( t\right) $, we have the \textit{quantum It\={o} product rule}
\begin{eqnarray}
d\left( X_{t}.Y_{t}\right) =dX_{t}.Y_{t}+X_{t}.dY_{t}+dX_{t}.dY_{t},
\label{Ito formula}
\end{eqnarray}
with the It\={o} correction given by 
\begin{eqnarray}
dX_{t}.dY_{t}&=& \sum_{k \in \mathsf{k}} \bigg\{ x_{0 k}\left( t\right) y_{k0 }\left( t\right) \,dt
+\sum_{j\in \mathsf{k}} x_{j k}\left( t\right) y_{k0 }\left( t\right) \,dB_j^\ast \left( t\right)
+\sum_{l\in \mathsf{k} } x_{0 k}\left( t\right) y_{kl }\left( t\right) \,dB_l \left( t\right)
+\sum_{j,l \in \mathsf{k} }x_{j k}\left( t\right) y_{kl }\left( t\right) \,d\Lambda_{jl}^\ast \left( t\right) \bigg\}
\nonumber \\
&=& x_{\alpha k}\left( t\right) y_{k\beta }\left( t\right) \,dB^{\alpha \beta }\left( t\right) .  
\label{Ito correction}
\end{eqnarray}
The quantum It\={o} table of gives the non-vanishing products of increments
\begin{eqnarray*}
d\Lambda _{ij}d\Lambda _{kl} &=&\delta _{jk}d\Lambda _{il},\qquad d\Lambda
_{ij}dB_{k}^{\ast }=\delta _{jk}dB_{i}^{\ast } \\
dB_{i}d\Lambda _{kl} &=&\delta _{ik}dB_{l},\qquad dB_{i}dB_{k}^{\ast
}=\delta _{ij}dt.
\end{eqnarray*}

For a product of $n$ quantum stochastic integrals, we get a sum of $2^{n}-1$
terms. For instance, triple products have the formula
\begin{eqnarray}
d\left( XYZ\right)  &=&\left( dX\right) YZ+X\left( dY\right) Z+XY\left(
dZ\right) \nonumber  \\
&&+\left( dX\right) \left( dY\right) Z+\left( dX\right) Y\left( dZ\right)
+X\left( dY\right) \left( dZ\right)  \nonumber\\
&&+\left( dX\right) \left( dY\right) \left( dZ\right) .
\label{eq:triple}
\end{eqnarray}

The general form of the constant operator-coefficient quantum stochastic
differential equation for an adapted unitary process $U$ is 
\begin{eqnarray}
dU\left( t\right) = \bigg\{ -\left( \frac{1}{2}L_{\mathsf{k}}^{\ast }L_{%
\mathsf{k}}+iH\right) dt+\sum_{j\in \mathsf{k}}L_{j}dB_{j}\left( t\right)
^{\ast }  -\sum_{j,k\in \mathsf{k}} L_j^\ast S_{jk} dB_{k}\left( t\right) +\sum_{j,k\in 
\mathsf{k}}(S_{jk}-\delta _{jk})d\Lambda _{jk}\left( t\right) \bigg\} %
U\left( t\right)
\label{eq:Ito_QSDE}
\end{eqnarray}
where the $S_{jk},L_{j}$ and $H$ are operators on the initial Hilbert space
which we collect together as
\begin{eqnarray}
S_{\mathsf{kk}}=\left[ 
\begin{array}{ccc}
S_{11} & \cdots & S_{1n} \\ 
\vdots & \ddots & \vdots \\ 
S_{n1} & \cdots & S_{nn}
\end{array}
\right] , \quad L_{\mathsf{k}}=\left[ 
\begin{array}{c}
L_{1} \\ 
\vdots \\ 
L_{n}
\end{array}
\right] .
\end{eqnarray}
The necessary and sufficient conditions for the process $U$ to be unitary are that $S_{\mathsf{kk}}=\left[ S_{jk}\right] _{j,k\in \mathsf{k}}$ is unitary (i.e., $\sum_k S_{kj}^\ast S_{kl} = \sum_k S_{jk}S_{lk}^\ast = \delta_{jl} \, I$)
and $H$ self-adjoint. (We use the convention that $L_{\mathsf{k}}=\left[
L_{k}\right] _{k\in \mathsf{k}}$ and that \ $L_{\mathsf{k}}^{\ast }L_{%
\mathsf{k}}=\sum_{k\in \mathsf{k}}L_{k}^{\ast }L_{k}$.) 
The triple $\left( S,L,H\right) $ are termed the \textit{Hudson-Parthasarathy parameters} of the open system evolution,
or more colloquially the $SLH$-\textit{coefficients}.

We may write the equation (\ref{eq:Ito_QSDE}) as $dU\left( t\right) =dG\left( t\right)
\,U\left( t\right) $ and we note that
\begin{eqnarray*}
dG\left( t\right) =
-\left( \frac{1}{2}L_{\mathsf{k}}^{\ast }L_{%
\mathsf{k}}+iH\right) dt+\sum_{j\in \mathsf{k}}L_{j}dB_{j}\left( t\right)
^{\ast }  -\sum_{j,k\in \mathsf{k}}S_{jk}L_{k}dB_{k}\left( t\right) +\sum_{j,k\in 
\mathsf{k}}(S_{jk}-\delta _{jk})d\Lambda _{jk}\left( t\right) .
\label{eq:dG}
\end{eqnarray*}
The isometry condition is $U^{\ast }\left( t\right) U\left( t\right) =I$
which in differential form is $\left( dU^{\ast }\right) U+U^{\ast }\left(
dU\right) +\left( dU^{\ast }\right) \left( dU\right) =0$, or equivalently
\begin{eqnarray*}
dG+dG^{\ast }+\left( dG^{\ast }\right) \left( dG\right) =0.
\label{eq:isometry}
\end{eqnarray*}
Likewise the co-isometry condition is $dG+dG^{\ast }+\left( dG\right) \left(
dG^{\ast }\right) =0$.

Let $X$ be an operator on the initial space, then we set $j_{t}\left(
X\right) \triangleq U\left( t\right) ^{\ast }\left[ X\otimes I\right]
U\left( t\right) $ to give its Heisenberg evolution. We refer to $\left\{
j_{t}\left( \cdot \right) :t\geq 0\right\} $ as the \textit{quantum stochastic flow}.
From the quantum It\={o} calculus, we have
\begin{eqnarray}
dj_{t}\left( X\right) =j_{t}\left( \mathscr{L}X\right)
dt+\sum_{i}j_{t}\left( \mathscr{M}_{i}X\right) dB_{i}^{\ast }\left( t\right)
+\sum_{i}j_{t}\left( \mathscr{N}_{i}X\right) dB_{i}\left( t\right)
+\sum_{j,k}j_{t}\left( \mathscr{S}_{jk}X\right) d\Lambda _{jk}\left(
t\right) ;
\label{eq:heisenberg}
\end{eqnarray}
where 
\begin{eqnarray}
\mathscr{L}X &=&\frac{1}{2}\sum_{i}L_{i}^{\ast }\left[ X,L_{i}\right] +\frac{%
1}{2}\sum_{i}\left[ L_{i}^{\ast },X\right] L_{i}-i\left[ X,H\right] , \quad 
\text{(the Lindbladian!)},\nonumber \\
\mathscr{M}_i X &=& \sum_j S_{ji}^\ast [X,L_j ], \nonumber \\
\mathscr{N}_i X &= & \sum_k [L_k^\ast , X ] S_{ki},\nonumber \\
\mathscr{S}_{ik} X &=& \sum_j S_{ji}^\ast XS_{jk} - \delta_{ik} X .
\label{eq:EH}
\end{eqnarray}

The equation (\ref{eq:heisenberg}) gives the dynamical evolution of the observable $j_t (X)$ and effectively
gives the dynamics of the system driven by the external inputs. We may also include outputs by setting
\begin{eqnarray}
B^{\text{out}}_k (t) \triangleq U ( t)^\ast \left[ I\otimes B_k (t) \right] U\left( t\right) .
\label{eq:output}
\end{eqnarray}
We use (\ref{eq:triple}) to compute the differential of $B^{\text{out}}_k (t)$.

From the quantum It\={o} table we have that two of the terms vanish straight
away: these are $\left( dU^{\ast }\right) \left( dB_{j}\right) U$ and $%
\left( dU^{\ast }\right) \left( dB_{j}\right) \left( dU\right) $, leaving 
\begin{eqnarray*}
dB_{j}^{\text{out}}=U^{\ast }\left( dB_{j}\right) U+\bigg\{ \left( dU^{\ast
}\right) B_{j}U+U^{\ast }B_{j}\left( dU\right) +\left( dU^{\ast }\right)
B_{j}\left( dU\right) \bigg\} +U^{\ast }\left( dB_{j}\right) \left(
dU\right) .
\label{dB_out}
\end{eqnarray*}
The first term here is readily seen to be $U^{\ast }\left( dB_{j}\right)
U\equiv dB_{j}$ since the increment is future pointing and commutes with
adapted coefficients. The next three terms can be rewritten as 
\begin{eqnarray*}
\bigg\{ \left( dU^{\ast }\right) B_{j}U+U^{\ast }B_{j}\left( dU\right)
+\left( dU^{\ast }\right) B_{j}\left( dU\right) \bigg\}  &=&U^{\ast
}\bigg\{ \left( dG^{\ast }\right) B_{j}+B_{j}\left( dG\right) +\left(
dG^{\ast }\right) B_{j}\left( dG\right) \bigg\} U \\
&=&U^{\ast }\bigg\{ \left( dG^{\ast }\right) +\left( dG\right) +\left(
dG^{\ast }\right) \left( dG\right) \bigg\} B_{j}U
\end{eqnarray*}
where we use the fact that $B_{j}\left( t\right) $ acts trivially on the
initial space and the future (beyond time $t$) factor of the Fock space
while the increments $dG$ and $dG^{\ast }$ in (\ref{eq:dG}) act non-trivially on these
factors, so that $dG\left( t\right) $ and $dG\left( t\right) ^\ast $
commute with $B_{j}\left( t\right) $ for each $t$. By the isometry condition
(\ref{eq:isometry}), these three terms vanish identically. The remaining term is
\begin{eqnarray*}
U^{\ast }\left( dB_{j}\right) \left( dU\right) =U^{\ast }\left(
dB_{j}\right) \left( dG\right) U=U^{\ast }\left( \sum_{k}\left(
S_{jk}-\delta _{jk}\right) dB_{jk}\left( t\right) +L_{k}dt\right)
U=\sum_{k}j_{t}\left( S_{jk}-\delta _{jk}\right) dB_{jk}\left( t\right)
+j_{t}\left( L_{k}\right) dt
\end{eqnarray*}
and so we obtain
\begin{eqnarray}
dB_{j}^{\text{out}}\left( t\right) =\sum_{k}j_{t}\left( S_{jk}\right)
\,dB_{k}\left( t\right) +j_{t}\left( L_j\right) \,dt .
\end{eqnarray}
The output field may then be measured. For instance, in a homodyne measurement we may
measure the quadrature process $Y_k = B_k^{\text{out}} +B^{\text{out} \ast}_k$. Note that in this case,
\begin{eqnarray}
dY_k (t) = \sum_{k}j_{t}\left( S_{jk}\right) \,dB_{k}\left( t\right) +\sum_{k}j_{t}\left( S_{jk}^\ast \right) \,dB_{k}^\ast\left( t\right) 
+j_{t}\left( L_j+ L_j^\ast \right) \,dt .
\label{eq:standard_quadrature_output}
\end{eqnarray}

\subsection{Controlled flows}
In principle, there is nothing to stop us replacing the $SLH$ coefficients in (\ref{eq:Ito_QSDE})
with adapted stochastic processes $\big\{ S_{jk}(t) ,L_{j}(t), H(t) : t\ge 0 \big\}$. Mathematically, the conditions that
$[ S_{jk} (t) ]_{j,k\in \mathsf{k}}$ be unitary and that $H(t)$ be self-adjoint, for all $t\ge 0$, are enough to ensure 
the unitarity of the corresponding process $\{ U(t) : t \ge 0 \}$.

To see why we might want to consider such models, let us look at the situation where the $S$ and $L$ are fixed operators 
on the initial space, and where the Hamiltonian is allowed to depend on a time-varying process $Z=\{ Z_t : t \geq 0\}$, say
\begin{eqnarray}
H(t) \equiv h (Z_t ) \, F ,
\end{eqnarray}
where $F$ is a fixed self-adjoint operator on the initial space, and $h (\cdot )$ is some real-valued function.

The corresponding unitary process may be denoted as $U_t^{[Z]}$ and may be referred to as a \textit{controlled unitary},
specifically controlled by the process $Z$. By construction, the controlled unitary depends on the control $Z$ in a causal
manner: $U_t^{[Z]}$ depends on the $\{ Z_s : 0 \le s \le t \}$ and not on values of $Z$ later than time $t$. We may take $Z$ 
to be a deterministic control function, however, more generally we could take it to be itself a quantum stochastic process, see Figure \ref{fig:NM_QFN_II_Fb}.
For instance, we may imagine performing a continuous measurement on the output process, and feed the measured output $Y$
back in as the control.

In \cite{Wiseman}, Wiseman considered direct feedback models where a formal Hamiltonian
correction $H_{\text{fb}}\left( t\right) =F\otimes \dot{Y}\left( t\right) $
was included. Here $\dot{Y}\left( t\right) $ is the formal derivative of the
measurement output process - as $Y$ is a diffusion process, and consequently
of unbounded variation, $\dot{Y}$ needs to be interpreted with care. A
rigorous way of interpreting Wiseman's Hamiltonian was present in \cite{GouJam09b} and
involves a double pass through the system - first corresponding to $SLH$%
-coefficients $\left( I,L,H\right) $ and second corresponding to $\left(
I,-iF,0\right) $. This can be viewed before feedback as the two-input
two-output device with $SLH$-coefficients $\left( \left[ 
\begin{array}{cc}
I & 0 \\ 
0 & I
\end{array}
\right] ,\left[ 
\begin{array}{c}
L \\ 
-iF
\end{array}
\right] ,H\right) $, that is,
\begin{eqnarray*}
dU\left( t\right) =\left\{ LdB_{1}\left( t\right) ^{\ast }-L^{\ast
}dB_{1}\left( t\right) -iF\left[ dB_{2}\left( t\right) +dB_{2}\left(
t\right) ^{\ast }\right] -\left( \frac{1}{2}L^{\ast }L+F^{2}+iH\right)
dt\right\} U\left( t\right) .
\end{eqnarray*}

\begin{figure}[h]
\centering
	\includegraphics[width=0.250\textwidth]{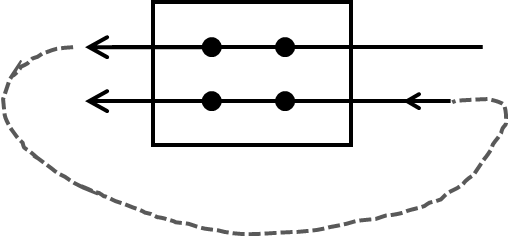}
	\caption{A double-pass (Series Product configuration) leading to Wiseman's direct feedback.}
	\label{fig:NM_QFN_II_Series}
\end{figure}

The first output is then fed in as the second input as shown in Figure \ref{fig:NM_QFN_II_Series}
to create a feedback loop. The term $F\left[ dB_{2}\left( t\right)
+dB_{2}\left( t\right) ^{\ast }\right] $ is then interpreted as realizing
the formal expression $H_{\text{fb}}\left( t\right) dt$. The closed loop
system then has $SLH$-coefficients given by the series product
\begin{eqnarray*}
\left( I,-iF,0\right) \vartriangleleft \left( I,L,H\right) =\left( I,L-iF,H+%
\frac{1}{2}\left( FL+L^{\ast }F\right) \right) 
\end{eqnarray*}
that is 
\begin{eqnarray*}
dU_{\text{fb}}\left( t\right) =\left\{ (L-iF)dB\left( t\right) ^{\ast
}-(L-iF)^{\ast }dB\left( t\right) -\left( \frac{1}{2}(L-iF)^{\ast }(L-iF)+iH+%
\frac{i}{2}\left( FL+L^{\ast }F\right) \right) dt\right\} U_{\text{fb}%
}\left( t\right) .
\end{eqnarray*}
The reduced model obtained this way agrees with the one derived by Wiseman.

The mathematical principle for dealing with this is quite subtle and is due to Bouten and van Handel \cite{BvH,BvH_ref}.
Their central observation was that one had to take care to distinguish the control process, $Z=B+B^\ast$, 
which is \textit{fed in} as a control
modifying the Hamiltonian $H$ - more generally, the $S$ and $L$ as well - and the output process $Y=B^{\text{out}}
+B^{\text{out} \ast}$ which is the measured output that you want to somehow feed back in. The two are related by
\begin{eqnarray}
Y(t) = U(t)^\ast  Z(t) U(t) .
\label{eq:Y_Z}
\end{eqnarray}
The subtlety is that $U$ is supposed now to be the modified $U_t^{[Z]}$ controlled by $Z$. In Wiseman's feedback model,
the modulator takes $Y$ and converts it into $W=\dot Y$ - it is described as proportional feedback, but in these terms
it is arguably more derivative controller than a proportional controller \cite{argument}.

This brings us to the main question which we aim to resolve in this paper: when to use $Z$ and when to use $Y$? In Wiseman's
derivation \cite{Wiseman}, the feedback Hamiltonian is $H_{\text{fb}}\left( t\right) =F\otimes \dot{Y}\left( t\right) $, 
however, in the Bouten and van Handel papers \cite{BvH,BvH_ref}, one constructs a controlled flow with $Z$ as the dependent process. In the derivation of Gough and James \cite{GouJam09b}, there is no explicit measurement - instead there is a second pass which has the coupling operator $-iF$ with $F=F^\ast$ and this some how picks out the quadrature $Z_2=B_2+B_2^\ast$, but after the series product is used to make the feedback connection everything is evolving coherently. As we shall see, there is an input picture and an output picture, and the two are unitarily equivalent. For model building, the input picture is preferable and it is here that we can apply the various interconnection rules for quantum feedback networks \cite{GouJam09a}. However, the controlled quantum flow can 
equally well described in the output picture, using $Y$, and the associated dynamical equations here are arguably more physically
intuitive.

\section{The Input and Output Pictures for Controlled Open Dynamics}
\label{sec:ConDyn}
Let us recall the various algebras we are considering in this theory
\begin{eqnarray*}
\mathfrak{A}_{0} &=&\mathscr{B}\left( \mathfrak{h}_{0}\right) \text{ \ the initial algebra
(system operators at time }t=0\text{)} \\
\mathfrak{N} &=&\mathscr{B}\left( \Gamma \left( \mathfrak{K}\otimes L^{2}[0,\infty )\right)
\right) \text{, \ the noise space.}
\end{eqnarray*}
We note that the continuous tensor product decomposition (\ref{eq:ctpd}) implies that $\mathfrak{N}=\mathfrak{N}_{[0,t]}\otimes \mathfrak{N}_{\left( t,\infty \right) }$, where 
$\mathfrak{N}_{\left[ 0,t\right] }=\mathscr{B}\left( \Gamma \left( \mathfrak{K}\otimes
L^{2}[0,t]\right) \right) $, etc. The algebra of all operators up to time $t$%
, (\ref{eq:A_t]}), is then 
\begin{eqnarray*}
\mathfrak{A}_{t]}=\mathfrak{A}_{0}\otimes \mathfrak{N}_{\left[ 0,t\right] }.
\end{eqnarray*}

\begin{definition}
A control process is a commutative family $Z=\left\{ Z_{t}:t\geq 0\right\} $
of adapted processes acting trivially on the initial space. The filtration determined by a
controlled process $Z$ is the collection of commutative von Neumann algebras
\begin{eqnarray*}
\mathfrak{Z}_{[0,t]}=\left\{ Z_{s}:0\leq s\leq t\right\} \subset \mathfrak{N}%
_{[0,t]},
\end{eqnarray*}
for $t >0$. 
\end{definition}

\begin{definition} A controlled stochastic process $F\left[ \left[ Z\right] \right] =\left\{ F_{t}\left[ %
\left[ Z\right] \right] :t\geq 0\right\} $, controlled by control process $Z$ is a quantum stochastic process
such that $F\left[ \left[ Z\right] \right]$ is affiliated with $\mathfrak{A}_0 \otimes \mathfrak{Z}_{[0,t]}$.
\end{definition}

\begin{definition}
For a fixed control process $Z$, let $\left( S\left[ \left[ Z\right] \right]
,L\left[ \left[ Z\right] \right] ,H\left[ \left[ Z\right] \right] \right) $
be controlled processes with $S_ t\left[ \left[ Z\right] \right] $ and $H_{t} [[ Z] ] $ unitary and self-adjoint, respectively, for
each $t$.  The unitary process they generate, $U_t^{[Z]} $, is the solution of the quantum stochastic differential
equation
$dU_t^{[Z]} =dG_t^{[Z]} \,  U_t^{[Z]} $ where
\begin{eqnarray}
dG_t^{[Z]} &=&
-\left( \frac{1}{2}L_{\mathsf{k},t}[[Z]]^{\ast }L_{\mathsf{k},t}[[Z]]+iH_t [[Z]] \right) dt
+\sum_{j\in \mathsf{k}}L_{j,t}[[Z]]dB_{j}\left( t\right)^{\ast }\nonumber \\
&&  -\sum_{j,k\in \mathsf{k}}L_{j,t} [[Z]]^\ast S_{jk,t}[[Z]]dB_{k}\left( t\right) +\sum_{j,k\in 
\mathsf{k}}(S_{jk,t}[[Z]]-\delta _{jk})d\Lambda _{jk}\left( t\right) ,
\label{eq:dGZ}
\end{eqnarray}
with $U(0)=I$.
\end{definition}

In our definition of controlled flows, we required
that $S_{jk,t}\left[ \left[ Z\right] \right] ,L_{k,t}\left[ \left[ Z\right] %
\right] $ and $H_{t}\left[ \left[ Z\right] \right] $ are controlled
processes $F_{t}\left[ \left[ Z\right] \right] $ - that is, they are
affiliated to $\mathfrak{A}_{0}\otimes \mathfrak{Z}_{[0,t]}$ which is a proper
subset of $\mathfrak{A}_{t]}$. The unitary process $U_{t}^{\left[ Z\right] }$
they generate however is not a controlled process. It is a quantum
stochastic integral whose coefficients are controlled process but the nature
of the integration will typically lead to a process in $\mathfrak{A} _{0}\otimes \mathfrak{N}_{[0,t]}$,
and outside $\mathfrak{A} _{0}\otimes \mathfrak{Z}_{[0,t]}$ in particular.
The same applies to $G_t^{[Z]}$ defined by (\ref{eq:dGZ}) with initial condition $G_0^{[Z]}=0$.

The identities $dG_t^{[Z] \, \ast} +dG_t^{[Z] \, \ast}+\left( dG_t^{[Z]   }\right) \left( dG_t^{[Z] \, \ast}\right) =0=
dG_t^{[Z] \, \ast}+dG_t^{[Z] \, \ast}+\left( dG_t^{[Z] \, \ast}\right) \left(dG_t^{[Z] }\right) $ still hold, and the unitary of $U_t^{[Z]}$ follows from this.

Again, let $X$ be an operator on the initial space, then we now set $j_{t}^{[Z]} \left(
X\right) \triangleq U_t^{[Z] \, \ast} \left[ X\otimes I\right]
U_t^{[Z]} $ to give its Heisenberg evolution. We refer to $\left\{
j_{t}^{[Z]}\left( \cdot \right) :t\geq 0\right\} $ as a \textit{controlled quantum stochastic flow}.
From the quantum It\={o} calculus, we have
\begin{eqnarray}
dj_{t}^{[Z]}\left( X\right) = U_t^{[Z] \, \ast } \bigg( \mathscr{L}_t^{[[Z]]}X \bigg) U_t^{[Z]} \,
dt+\sum_{i} U_t^{[Z]\, \ast } \bigg( \mathscr{M}_{i,t}^{[[Z]]}X \bigg) U_t^{[Z]} \, dB_{i}^{\ast }\left( t\right) \nonumber \\
+\sum_{i} U_t^{[Z] \, \ast} \bigg(\mathscr{N}_{i,t}^{[[Z]]}X \bigg) U_t ^{[Z]} \, dB_{i}\left( t\right)
+\sum_{j,k}U_t^{[Z] \, \ast } \bigg( \mathscr{S}_{jk,t}^{[[Z]]}X \bigg) U_t ^{[Z]} \, d\Lambda _{jk}\left(
t\right) ;
\label{eq:heisenberg_Z}
\end{eqnarray}
where the new super-operators $\mathscr{L}_t^{[[Z]]}, \mathscr{M}_{i,t}^{[[Z]]}, \mathscr{N}^{[[Z]]}_{i,t}
, \mathscr{S}^{[[Z]]}_{jk,t}$
are just the same for as in (\ref{eq:EH}) with the $SLH$-coefficients now replaced by the controlled versions
$S_t [[Z]], L_t [[Z]] , H_t [[Z]]$.

Likewise, the output processes are now $B^{\text{out},[Z]}_k (t) \triangleq U_t^{[Z]\, \ast } \left[ I\otimes B_k (t) \right] 
U_t^{[Z]}$. We may again use the quantum It\={o} calculus as before to derive the analogue of (\ref{dB_out}).
We find that the first term should be 
$U_t ^{[Z]\, \ast }\,dB_{j} (t) \, U_t^{[Z]}$ which again equals $dB_{j}(t)$. However, the argument we previously used
to show that the next group of three terms vanishes breaks down since $B_j (t)$ no longer commutes with $dG_t^{[Z]}$
due to its possible dependence on $Z_s$ for $0\le s \le t$.

Instead we obtain
\begin{eqnarray}
dB_j^{\text{out}}( t) &=&\sum_k U_t^{[Z] \, \ast}\big( S_{jk,t}[[Z]] \big) U_t^{[Z]}
\,dB_k ( t) + U_t^{[Z] \, \ast} \big( L_j \big)U^{[Z]}_t \, dt  \nonumber \\
&& + U_t^{[Z] \, \ast} \bigg\{ \left( dG_t^{[Z] \,\ast }\right) B_{j}(t) +B_{j}(t) \left( dG_t^{[Z]}\right) +\left(
dG^{[Z] \ast }_t  \right) B_{j}(t) \left( dG_t^{[Z]}\right) \bigg\}U^{[Z]}_t.
\label{eq:B_out_Z}
\end{eqnarray}
The term in braces in (\ref{eq:B_out_Z}) may be written as
\begin{eqnarray}
&&\Bigg\{
\frac{1}{2}\sum_i L_{i,t}[[Z]]^{\ast }\bigg[ B_j (t),L_{i,t}[[Z]]\bigg] +\frac{%
1}{2}\sum_{i}\bigg[ L_{i,t}[[Z]]^{\ast },B_j (t)\bigg] L_{i,t}[[Z]]-i\bigg[ B_j (t),H_t [[Z]]\bigg] \Bigg\} dt \nonumber \\
&&+\sum_{i,k}S_{ki,t}[[Z]]^\ast \bigg[ B_j (t) ,L_{k,t}[[Z]] \bigg] \,  dB_{i}^{\ast }\left( t\right) \nonumber \\
&&+\sum_{i,k} \bigg[ L_{k,t} [[Z]]^\ast , B_j (t) \bigg] S_{ki,t} [[Z]] \, dB_{i}\left( t\right) \nonumber \\
&& +\sum_{i,k}  \bigg( \sum_l S_{li,t}[[Z]]^\ast B_j (t) S_{lk,t} [[Z]] - \delta_{ik} B_j (t)\bigg) \,  d\Lambda _{ik}\left(
t\right).
\label{eq:braces}
\end{eqnarray}
Previously this vanishes since  the $B_j (t)$ fields lived
in the noise algebra while the $SLH$ coefficients lived in the initial algebra - this time the $SLH$ coefficients are themselves adapted processes.

\subsection{The Control and the Measurement Algebra}

Given an adapted unitary process $U\left( \cdot \right) $ we obtain the
measurement algebras
\begin{eqnarray*}
\mathfrak{Y}_{[0,t]}=U\left( t\right) ^{\ast }\mathfrak{Z}_{[0,t]}U\left( t\right)
=\left\{ Y_{s}:0\leq s\leq t\right\} 
\end{eqnarray*}
where $Y_{t}=U\left( t\right) ^{\ast }Z_{t}U\left( t\right) $. (Note that $%
Y_{s}\equiv U\left( t\right) ^{\ast }Z_{s}U\left( t\right) $ whenever $t\geq
s$.) The algebras $\mathfrak{Y}_{[0,t]}$ are again commutative, and form a
filtration of von Neumann algebras.

\bigskip 

\begin{observation} We have introduced controlled processes $F_{t}\left[ %
\left[ Z\right] \right] $ which, for each fixed $t$, is an operator taking
values in $\mathfrak{A}_{0}\otimes \mathfrak{Z}_{[0,t]}$ with commutative domain $%
\mathfrak{Z}_{[0,t]}$. A central feature is that the domain space commutes with
the value space. Now let $U\left( \cdot \right) $ be an adapted unitary
process and define
\begin{eqnarray*}
\tilde{F}_{t}\left[ \left[ Y\right] \right] =U\left( t\right) ^{\ast }\,F_{t}%
\left[ \left[ Z\right] \right] \,U\left( t\right) .
\end{eqnarray*}
Let us set $\mathfrak{A}_{t}=U\left( t\right) ^{\ast }\mathfrak{A}_{0}U\left(
t\right) $ which is the initial algebra unitarily rotated by $U\left(
t\right) $, and then $\mathfrak{N}^{\text{out}}_{[0,t]}=U\left( t\right) ^{\ast }%
\mathfrak{N}_{[0,t]}U\left( t\right) $ which is the noise algebra unitarily
rotated by $U\left( t\right) $ the total output algebra. Note that the
measurement algebra $\mathfrak{Y}_{[0,t]}$ is a (commutative) subalgebra of $%
\mathfrak{N}_{[0,t]}^{\text{out}}$. The rotated process $\tilde{F}_{t}\left[ %
\left[ Y\right] \right] $ takes values in $U\left( t\right) ^{\ast }\mathfrak{A}%
_{t]}U\left( t\right) \cong \mathfrak{A}_{t}\otimes \mathfrak{Y}_{[0,t]}$ and has
domain $\mathfrak{Y}_{[0,t]}$, and again we have the central feature is that the
domain space commutes with the value space. This, in fact, is our justification for writing it in
the form $\tilde{F}_{t}\left[ \left[ Y\right] \right] $.
\end{observation}

\begin{notation}
We now take the unitary process $U$ performing the rotation to be the controlled unitary $U^{[Z]}$.
In the following, we will assume this choice, and set
\begin{eqnarray}
\tilde{F}_{t}\left[ \left[ Y\right] \right] \triangleq U_t^{[Z] \, \ast } \,F_{t}%
\left[ \left[ Z\right] \right] \,U_t^{[Z]} .
\label{eq:conversion}
\end{eqnarray}
For the special case where $F_t [[Z]] \equiv X \otimes I$, with $X \in \mathfrak{A}_0$, we will also write
\begin{eqnarray}
j_{t} ^{\left[ \left[ Y\right] \right] } (X) \triangleq U_t^{[Z] \, \ast } \, \big[ X \otimes I \big] \,U_t^{[Z]} ,
\label{eq:j_t^Y}
\end{eqnarray}
for $F_t [[Z]] \equiv I \otimes B(t)$ we write
\begin{eqnarray}
B^{\left[ \left[ Y\right] \right] \text{out}}_k(t) \triangleq U_t^{[Z] \, \ast } \, \big[ I \otimes B_k(t) \big] \,U_t^{[Z]} ,
\label{eq:B^Y}
\end{eqnarray}
and similarly 
$\Lambda_{jk}^{\left[ \left[ Y\right] \right] \text{out}}(t) \triangleq 
U_t^{[Z] \, \ast } \, \big[ I \otimes \Lambda_{jk}(t) \big] \,U_t^{[Z]}$
We refer top the description in terms of the control process $Z$ as the input picture and the description in terms of 
the measurement readout process $Y$ as the output picture.
\end{notation}

Note that (\ref{eq:j_t^Y}) and (\ref{eq:B^Y}) were previously denoted as $j_t^{[Z]} (X)$ and $B^{\left[ Y\right] \text{out}}_k(t)$,
respectively. The reader may well have noticed that sometimes we use single square brackets and sometimes double - to recap, we use single brackets to indicate a dependence of an operator-valued process on a commutative control process but use 
double square brackets to indicate an operator-valued function with a commutative domain whenever the domain and range variables commute. 

\begin{proposition}[Non-demolition Principle for Controlled Flows]
For each $s \le t$ and $X \in \mathfrak{A}_0$, we have that $Y_s$ will 
commute with $j_{t} ^{\left[ \left[ Y\right] \right] } (X)$.
\end{proposition}
\begin{proof} We define the two-parameter family $\big\{ U_{t,s}^{[Z]}: t \ge s \big\}$ by $U_{t,s}^{[Z]} =
I + \int_s^t dG^{[Z]}_\tau U^{[Z]}_{\tau , s} $. We see that $ U^{[Z]}_t \equiv U^{[Z]}_{t,0}$, and that
$U^{[Z]}_t \equiv U^{[Z]}_{t,s} \, U^{[Z]}_s$. In particular, $Z_s$ commutes with $U^{[Z]}_{t,s}$ for $t \ge s$,
since $dG_\tau [[Z]]$ commutes with $Z_s$ for all $s \le \tau \leq t$: this implies that for $t \ge s$,
\begin{eqnarray*}
U^{[Z]\ast}_t \big[ I \otimes Z_s \big] U^{[Z]}_t = 
U^{[Z]\ast}_s U^{[Z]\ast}_{t,s}  \big[ I \otimes Z_s \big]  U^{[Z]}_{t,s}U^{[Z]}_s
= Y_s .
\end{eqnarray*}
Therefore, for $t \geq s$,
\begin{eqnarray*}
\big[ j_t^{[[Y]] }(X) , Y_s \big]  = U^{[Z]\ast}_t  \big[ X \otimes I, I \otimes Z_s \big]  U^{[Z]}_s =0.
\end{eqnarray*}
\end{proof} 

While the proof is similar to that for standard uncontrolled flows, see \cite{BouvanHJam07}, however
we note that in the latter situation we
also have the separate identities $\big[ j_t (X) ,B^{ \text{out}}_k(t)  \big] =
\big[ j_t (X) ,B^{  \text{out}\ast}_k(t)  \big] =
\big[ j_t  (X) ,\Lambda^{  \text{out}}_{jk}(t)  \big] =
0$, which need not necessarily hold true for controlled flows.

\begin{lemma}
Let us fix a control process $Z$ and take $U^{\left[ Z\right] }$ to be the unitary evolution generated by the
controlled $SLH$ coefficient processes $\left( S\left[ \left[ Z\right] \right]
,L\left[ \left[ Z\right] \right] ,H\left[ \left[ Z\right] \right] \right) $.
Then $U_{t}^{\left[ Z\right] }\equiv V_{t}^{\left[ Y\right] }$ where $V_{t}^{%
\left[ Y\right] }$ is the solution to the quantum stochastic differential
equation
\begin{eqnarray}
dV_{t}^{\left[ Y\right] }=V_{t}^{\left[ Y\right] }\,d\tilde{G}_{t} ^{
\left[ Y  \right]} ,\quad V_{0}^{\left[ Y\right] }=I,
\end{eqnarray}
where 
\begin{eqnarray}
d \tilde{G}_t^{[Y]}  &=&
-\left( \frac{1}{2} \tilde{L}_{\mathsf{k},t}[[Y]]^{\ast } \tilde{L}_{\mathsf{k},t}[[Y]]+i \tilde{H}_t [[Y]] \right) dt
+\sum_{j\in \mathsf{k}} \tilde{L}_{j,t}[[Y]]dB_{j}\left( t\right)^{\ast }\nonumber \\
&&  -\sum_{j,k\in \mathsf{k}} \tilde{S}_{jk,t}[[Y]] \tilde{L}_{k,t}[[Y]]dB_{k}\left( t\right) +\sum_{j,k\in 
\mathsf{k}}(\tilde{S}_{jk,t}[[Y]]-\delta _{jk})d\Lambda _{jk}\left( t\right) ,
\label{eq:dGY}
\end{eqnarray}
\end{lemma}

\begin{proof}
Whilst this is one of the key observations we need from a conceptual; point
of view, its proof is actually trivial. The quantum stochastic differential
equation for the unitary $U^{\left[ Z\right] }$ is easily rearranged to read as
\begin{eqnarray*}
dU_{t}^{\left[ Z\right] }=dG_{t}^{ \left[ Z\right] } \,U_{t}^{%
\left[ Z\right] }=U_{t}^{\left[ Z\right] }U_{t}^{\left[ Z\right] \ast }dG_{t}^{\left[ Z\right] }
 \,U_{t}^{\left[ Z\right] }=U_{t}^{\left[ Z\right] }\,d\tilde{G}_{t}^{ \left[ Y\right] } ,
\end{eqnarray*}
which is the same one as $V_{t}^{\left[ Y\right] }$ with the same initial
condition. 
\end{proof}

\subsection{The Controlled Stochastic Heisenberg Dynamics}

For a fixed operator $X$ in the initial algebra, we have introduced the
evolution $j_{t}^{\left[ Z\right] }\left( X\right) =U_{t}^{\left[ Z\right]
\ast }\left( X\otimes I\right) U_{t}^{\left[ Z\right] }$. This is clearly a
special case of (\ref{eq:conversion}) with $F_{t}\left[ \left[ Z\right] \right] =X\otimes I$.
We could therefore write $j_{t}^{\left[ Z\right] }\left( X\right) $ as $X_{t}%
\left[ \left[ Z\right] \right] $ or, equivalently, as $\tilde{X}_{t}\left[ %
\left[ Y\right] \right] $. The Heisenberg-Langevin equation (\ref{eq:heisenberg_Z}) then becomes

\begin{eqnarray}
d\tilde{X}_{t}\left[ \left[ Y\right] \right] &=& \Bigg\{\frac{1}{2}\sum_{i}\tilde{L}_{i,t}[[Y]]^{\ast }\bigg[\tilde{X}_{t}%
\left[ \left[ Y\right] \right] ),\tilde{L}_{i,t}[[Y]]\bigg]+\frac{1}{2}%
\sum_{i}\bigg[\tilde{L}_{i,t}[[Y]]^{\ast },\tilde{X}_{t}\left[ \left[ Y%
\right] \right] )\bigg]\tilde{L}_{i,t}[[Y]]-i\bigg[\tilde{X}_{t}\left[ \left[
Y\right] \right] ,\tilde{H}_{t}[[Y]]\bigg]\Bigg\}dt  \nonumber \\
&&+\sum_{i,k}\tilde{S}_{ki,t}[[Y]]^{\ast }\bigg[\tilde{X}_{t}\left[ \left[ Y%
\right] \right] ,\tilde{L}_{k,t}[[Y]]\bigg]\,dB_{i}^{\ast }\left( t\right)  
\nonumber \\
&&+\sum_{i,k}\bigg[\tilde{L}_{k,t}[[Y]]^{\ast },\tilde{X}_{t}\left[ \left[ Y%
\right] \right] \bigg]\tilde{S}_{ki,t}[[Y]]\,dB_{i}\left( t\right)   \nonumber
\\
&&+\sum_{i,k}\bigg(\sum_{l}\tilde{S}_{li,t}[[Y]]^{\ast }\tilde{X}_{t}\left[ %
\left[ Y\right] \right] (t)\tilde{S}_{lk,t}[[Y]]-\delta _{ik}\tilde{X}_{t}%
\left[ \left[ Y\right] \right] \bigg)\,d\Lambda _{ik}\left( t\right) .
\label{eq:Heis_Y}
\end{eqnarray}

\subsection{The Controlled System Input-Output Relations}

The output is given as $B_{j}^{\text{out}\left[ Z\right] }\left( t\right)
=U_{t}^{\left[ Z\right] \ast }\left( I\otimes B_{j}\left( t\right) \right)
U_{t}^{\left[ Z\right] }$. Let us write this as $\tilde{B}_{j}^{\text{out}%
\left[ Y\right] } (t)$, then from (\ref{eq:B_out_Z}) and (\ref{eq:braces}) we have
\begin{eqnarray}
d\tilde{B}_{j}^{\text{out}\left[ Y\right] } (t)&&= \sum_{k}\tilde{S}_{jk,t}\left[
\left[ Y\right] \right] d B_{j}  (t) +\tilde{L}%
_{j,t}\left[ \left[ Y\right] \right] dt  \nonumber \\
&&+\Bigg\{\frac{1}{2}\sum_{i}\tilde{L}_{i,t}[[Y]]^{\ast }\bigg[\tilde{B}%
_{j}^{\text{out}\left[ Y\right] }(t),\tilde{L}_{i,t}[[Y]]\bigg]+\frac{1}{2}%
\sum_{i}\bigg[\tilde{L}_{i,t}[[Y]]^{\ast },\tilde{B}_{j}^{\text{out}\left[ Y%
\right] } (t )\bigg]\tilde{L}_{i,t}[[Y]]-i\bigg[%
\tilde{B}_{j}^{\text{out}\left[ Y\right] }(t),\tilde{H}_{t}[[Y]]\bigg]\Bigg\}dt
\nonumber
\\
&&+\sum_{i,k}\tilde{S}_{ki,t}[[Y]]^{\ast }\bigg[\tilde{B}_{j}^{\text{out}%
\left[ Y\right] } (t) ,\tilde{L}_{k,t}[[Y]]\bigg]\,dB_{i}^{\ast }\left( t\right) 
\nonumber \\
&&+\sum_{i,k}\bigg[\tilde{L}_{k,t}[[Y]]^{\ast },\tilde{B}_{j}^{\text{out}%
\left[ Y\right] } (t) \bigg]\tilde{S}_{ki,t}[[Y]]\,dB_{i}\left( t\right)   \nonumber
\\
&&+\sum_{i,k}\bigg(\sum_{l}\tilde{S}_{li,t}[[Y]]^{\ast }\tilde{B}_{j}^{\text{%
out}\left[ Y\right] }(t)\tilde{S}_{lk,t}[[Y]]-\delta _{ik}\tilde{B}_{j}^{%
\text{out}\left[ Y\right] }(t) \bigg)\,d\Lambda _{ik}\left( t\right) .
\label{eq:io_Y}
\end{eqnarray}

\section{Quadrature Feedback}
\label{sec:Quad}
Let us consider a homodyne measurement scheme where we measure the
quadrature process. Here we set
\begin{eqnarray*}
Z\left( t\right) =B\left( t\right) +B\left( t\right) ^{\ast }
\end{eqnarray*}
(for simplicity we consider only a single input single output model, $n=1$).
The measured output is then 
\begin{eqnarray}
Y\left( t\right) =U_{t}^{\left[ Z\right] \ast }\left[ B\left( t\right)
+B\left( t\right) ^{\ast }\right] U_{t}^{\left[ Z\right] }=\tilde{B}^{\text{%
out}}\left( t\right) +\tilde{B}^{\text{out}}\left( t\right) ^{\ast } .
\label{eq:Y_quad}
\end{eqnarray}

In the present case ($n=1$) the equations (\ref{eq:io_Y}) reduce to
\begin{eqnarray}
d\tilde{B}^{\text{out}}\left( t\right)  &=&\tilde{S}_{t}dB\left( t\right) +%
\tilde{L}_{t}dt  \nonumber \\
&&+\Bigg\{\frac{1}{2} \tilde{L}_{t}{}^{\ast }\bigg[\tilde{B}^{\text{%
out}   }(t),\tilde{L}_{t}\bigg]+\frac{1}{2} \bigg[\tilde{L}_{t}{}^{\ast
},\tilde{B}^{\text{out}}(t)\bigg]\tilde{L}_{t}-i\bigg[\tilde{B}^{\text{out}%
}(t),\tilde{H}_{t}\bigg]\Bigg\}dt  \nonumber \\
&&+\tilde{S}_{t}^{\ast }\bigg[\tilde{B}^{\text{out}}(t),\tilde{L}_{t}\bigg]%
\,dB\left( t\right)^\ast +\bigg[\tilde{L}_{t}{}^{\ast },\tilde{B}^{\text{out}%
}\left( t\right) \bigg]\tilde{S}_{t}\,dB\left( t\right) +\bigg(\tilde{S}%
_{t}{}^{\ast }\tilde{B}^{\text{out}}(t)\tilde{S}_{t}-\tilde{B}^{\text{out}%
}\left( t\right) \bigg)\,d\Lambda \left( t\right) .
\label{eq:dB_quad}
\end{eqnarray}
where we drop the $Y$-dependence for convenience.

When dealing with terms such as $\bigg[\tilde{B}^{\text{out} [Y]}(t),\tilde{H}_{t} [[Y]] \bigg]$, we would like to know whether we
remain in the algebra. We note that this may be written as $U_t^{[Z]} \bigg[ I\otimes B(t), H_t [[Z]] \bigg] U^{[Z]}_t$,
so an equivalent problem is to show that $\bigg[ I\otimes B(t), H_t [[Z]] \bigg]$ 
remains in $\mathfrak{A}_0 \otimes \mathfrak{Z}_{[0,t]}$.
The following Lemma gives an affirmative answer.

\begin{lemma}
Let $F_{t}\left[ \left[ Z\right] \right] $ be a stochastic process
controlled by the quadrature process $Z$. Then the commutator $\bigg[ I \otimes
B\left( t\right) ,F_{t}\left[ \left[ Z\right] \right] \bigg] $ is
affiliated to $\mathfrak{A}_0 \otimes \mathfrak{Z}_{[0,t]}$.
\end{lemma}

\begin{proof}
We may assume that the process $F_{t}\left[ \left[ Z\right] \right] $ admits
a chaotic expansion of the form
\begin{eqnarray*}
F_{t}\left[ \left[ Z\right] \right] =\sum_{n=0}^{\infty }\int_{\Delta
_{n}\left( t\right) }F\left( \tau _{n},\cdots ,\tau _{1}\right) \otimes dZ_{\tau
_{n}}\cdots dZ_{\tau _{1}}
\end{eqnarray*}
where $\Delta _{n}\left( t\right) $ is the simplex $\left\{ \left( \tau
_{n},\cdots ,\tau _{1}\right) :t\geq \tau _{n}>\cdots >\tau _{1}\geq
0\right\} $. The commutation relations for the creation and annihilation
processes are $\left[ B_{t},B_{s}^{\ast }\right] =t\wedge s\,I$, where $%
t\wedge s$ is the minimum of $t$ and $s$. We therefore see, that under the
integral sign,
\begin{eqnarray*}
\left[ B_{t},dZ_{\tau }\right] \equiv d\tau ,\quad \left( \tau <t\right) .
\end{eqnarray*}
Therefore
\begin{eqnarray*}
\big[ I \otimes B\left( t\right) ,F_{t}\left[ \left[ Z\right] \right] \big]
=\sum_{n=1}^{\infty }\sum_{k=1}^{n}\int_{\Delta _{n}\left( t\right) }F\left(
\tau _{n},\cdots ,\tau _{1}\right) \otimes dZ_{\tau _{n}}\cdots d\tau _{k}\cdots
dZ_{\tau _{1}}.
\end{eqnarray*}
We see that the commutator is evidently affiliated with $\mathfrak{A}_0 \otimes\mathfrak{Z}_{[0,t]}$.
\end{proof}

\begin{proposition}
\label{prop:io_quad}
The equation (\ref{eq:Y_quad}) implies that 
\begin{eqnarray}
dY\left( t\right) =\tilde{S}_{t}dB\left( t\right) +\tilde{S}_{t}^{\ast
}dB\left( t\right) ^{\ast }+(\tilde{L}_{t}+\tilde{L}_{t}^{\ast })dt.
\label{eq:dY_quad}
\end{eqnarray}
\end{proposition}

\begin{proof}
We now substitute (\ref{eq:dB_quad}) into $dY\left( t\right) =d\tilde{B}^{\text{out}}\left(
t\right) +d\tilde{B}^{\text{out}}\left( t\right) ^{\ast }$ - we find that the first two terms combine to give (\ref{eq:dY_quad}
while the other terms vanish. To see why, let us examine the $d\Lambda $
term, we will have
\begin{eqnarray*}
\bigg(\tilde{S}_{t}{}^{\ast }\tilde{B}^{\text{out}}(t)\tilde{S}_{t}-\tilde{B}%
^{\text{out}}\left( t\right) \bigg)+\bigg(\tilde{S}_{t}{}^{\ast }\tilde{B}^{%
\text{out}}(t)\tilde{S}_{t}-\tilde{B}^{\text{out}}\left( t\right) \bigg)%
^{\ast }\equiv \bigg(\tilde{S}_{t}{}^{\ast }Y(t)\tilde{S}_{t}-Y\left(
t\right) \bigg)
\end{eqnarray*}
but $\tilde{S}_{t}=\tilde{S}_{t}\left[ \left[ Y\right] \right] $ commutes
with $Y\left( t\right) $ by construction. Similarly, with the combination of $%
d\tilde{B}^{\text{out}}\left( t\right) $ and $d\tilde{B}^{\text{out}}\left(
t\right) ^{\ast }$, we reconstitute commutators of $Y\left( t\right) =\tilde{%
B}^{\text{out}}\left( t\right) +\tilde{B}^{\text{out}}\left( t\right) ^{\ast
}$ with $\tilde{L}_{t}\left[ \left[ Y\right] \right] $ and $\tilde{H}_{t}%
\left[ \left[ Y\right] \right] $, and these similarly vanish identically.
\end{proof}

\subsection{Feedback to the Hamiltonian}
We now look at several examples.

\subsubsection{Proportional Feedback}
Here we just take $\left( W_{t}=Z_{t}\right) $, and choose the $SLH$ coefficients to be
We set 
\begin{eqnarray*}
S_{t}\left[ \left[ Z\right] \right] =S,\quad L_{t}\left[ \left[ Z\right] %
\right] =L,\quad H_{t}\left[ \left[ Z\right] \right] =F\otimes Z_{t},
\end{eqnarray*}
where $S,L$ and $F=F^{\ast }$ are fixed operators on the system space $\mathfrak{%
h}_{0}$. We
find that 
\begin{eqnarray*}
dB^{\text{out}}\left( t\right) = j_t^{[[Y]]}(S) \,
dB\left( t\right) +j_t^{[[Y]]}(L-itF)  \, dt
\end{eqnarray*}
which follows from (\ref{eq:io_Y})\ and the observation that $\big[
I\otimes B\left( t\right) ,H_{t}\left[ \left[ Z\right] \right] \big]
=F\otimes tI$. From (\ref{eq:Heis_Y}), the Heisenberg equations are then
\begin{eqnarray}
d j_t^{[[Y]]} (X)  &=&\Bigg\{ j_t^{[[Y]]} \left( \mathscr{L}\left( X\right)  \right) 
-i \, j_t^{[[Y]]} \left( [X,F] \right)  \,Y_{t} \bigg]\Bigg\}dt  \nonumber \\
&&+ j_t^{[[Y]]} \big( \mathscr{M}\left( X\right) \big)\,dB^{\ast }\left(t\right) 
+ j_t^{[[Y]]} \big( \mathscr{N}\left( X\right) \big) \,dB^{\ast }\left(t\right) 
+ j_t^{[[Y]]} \big( \mathscr{S}\left( X\right) \big) \,d\Lambda^{\ast }\left(
t\right) .   \nonumber
\end{eqnarray}

\subsubsection{Nonlinear Modulator}
This time we replace the modulated process to be $\left( W_{t}=g\left( Z_{t}\right) \right) $, so that 
\begin{eqnarray*}
S_{t}\left[ \left[ Z\right] \right] =S,\quad L_{t}\left[ \left[ Z\right] %
\right] =L,\quad H_{t}\left[ \left[ Z\right] \right] =F\otimes g(Z_{t}),
\end{eqnarray*}
where now $g$ is some nonlinear function. This time we have $\big[ I\otimes
B\left( t\right) ,H_{t}\left[ \left[ Z\right] \right] \big] =F\otimes
tg^{\prime }\left( Z_{t}\right) $ and so 
\begin{eqnarray*}
dB^{\text{out}}\left( t\right) = j_t^{[[Y]]} (S) \,
dB\left( t\right) +\left( j_t^{[[Y]]} (L) -it j_t^{[[Y]]} (F) 
\, g^{\prime }\left( Y_{t}\right)
\right) dt.
\end{eqnarray*}
(Note that $g^{\prime
}\left( Y_{t}\right) $ will commute with $j_t^{[[Y]]} (F) $, etc.) From (\ref{eq:Heis_Y}), the Heisenberg equations are then
\begin{eqnarray}
d j_t^{[[Y]]} (X )  &=&\Bigg\{ j_t^{[[Y]]} \big( \mathscr{L}\left( X\right)\big)
-i j_t^{[[Y]]} \big( [X,F] \big)  \,g(Y_{t}) \bigg]\Bigg\}dt  \nonumber \\
&&+ j_t^{[[Y]]} \big( \mathscr{M}\left( X\right) \big)\,dB^{\ast }\left(t\right) 
+ j_t^{[[Y]]} \big( \mathscr{N}\left( X\right) \big) \,dB^{\ast }\left(t\right) 
+ j_t^{[[Y]]} \big( \mathscr{S}\left( X\right) \big) \,d\Lambda^{\ast }\left(
t\right) .   \nonumber
\end{eqnarray}

\subsubsection{Causal Linear Filter Modulator}

More generally, the modulator may act as a causal linear filter, say
\begin{eqnarray*}
S_{t}\left[ \left[ Z\right] \right] =S,\quad L_{t}\left[ \left[ Z\right] %
\right] =L,\quad H_{t}\left[ \left[ Z\right] \right] =\int_{0}^{t}F\left(
t-s\right) \otimes dZ_{s},
\end{eqnarray*}
where $F\left( \cdot \right) $ is a self-adjoint $\mathfrak{A}_{0}$-valued
function of time. (For the special case $F (t) = F \, h(t)$, with $F=F^\ast \in \mathfrak{A}_0$ and $h(\cdot )$ a fixed
real-valued function, we obtain $H_t [[Z]] = F \otimes W_t$ where $W_t =
\int_{0}^{t}h\left(
t-s\right) \otimes dZ_{s}$ is a convolution.)
This time we have $\big[ I\otimes B\left( t\right) ,H_{t}%
\left[ \left[ Z\right] \right] \big] =\int_{0}^{t}F\left( t-s\right)
\otimes ds\equiv \int_{0}^{t}F\left( u\right) \otimes du$, and so 
\begin{eqnarray*}
dB^{\text{out}} ( t ) =j_t^{[[Y]]}(S) \,
dB ( t ) +\left( j_t^{[[Y]]} (L)
-i\int_{0}^{t} j_t^{[[Y]]} (F_{u})  du \right) dt.
\end{eqnarray*}
From (\ref{eq:Heis_Y}), the Heisenberg equations are then
\begin{eqnarray}
dj_t^{[[Y]]} (X)    &=& \bigg\{  j_t^{[[Y]]} \big(\mathscr{L}\left( X\right) \big)
-i\int_{0}^{t} j_t^{[[Y]]} \big( [X,F_{t-s}] \big) \,dY_{s}\bigg] \bigg\} dt  \nonumber \\
&&+ j_t^{[[Y]]} \big( \mathscr{M}\left( X\right) \big)\,dB^{\ast }\left(t\right) 
+ j_t^{[[Y]]} \big( \mathscr{N}\left( X\right) \big) \,dB^{\ast }\left(t\right) 
+ j_t^{[[Y]]} \big( \mathscr{S}\left( X\right) \big) \,d\Lambda^{\ast }\left(
t\right) .  \nonumber
\end{eqnarray}
(Again, note that the integrand $j_t^{[[Y]]} \big( [X,F_{t-s}] \big) $ commutes with the increment $dY_{s}$.)

\subsection{Feedback to the Coupling Operator}

Let us consider a cavity mode $a$ with the $SLH$-coefficients
\begin{eqnarray*}
S_{t}\left[ \left[ Z\right] \right] =e^{i\theta },\quad L_{t}\left[ \left[ Z%
\right] \right] =\sqrt{\gamma }a\otimes I+\lambda I\otimes Z_{t},\quad
H=\omega a^{\ast }a.
\end{eqnarray*}
We denote the time-evolved mode as $\tilde{a}_{t}= j_t^{[[Y]]} (a) = U_{t}^{\left[ Z\right]
\ast }\left( a\otimes I\right) U_{t}^{\left[ Z\right] }$, then the
Heisenberg equations, and input-output equations read as
\begin{eqnarray*}
d\tilde{a}_{t} &=&-\frac{1}{2}\left( \sqrt{\gamma }\tilde{a}_{t}+\lambda
Y_{t}\right) dt-i\omega \tilde{a}_{t}-\sqrt{\gamma }dB\left( t\right)  \\
dY_{t} &=&e^{i\theta }dB\left( t\right) +e^{-i\theta }dB\left( t\right)
^{\ast }+\left[ \sqrt{\gamma }\left( \tilde{a}_{t}+\tilde{a}_{t}^{\ast
}\right) +2\lambda Y_{t}\right] dt
\end{eqnarray*}
which may be written as the \textit{linear} differential equation $dx_{t}=\mathsf{A}x_{t}dt+\mathsf{B}du_{t}$ where
\begin{eqnarray*}
x_{t}=\left[ 
\begin{array}{c}
\tilde{a}_{t} \\ 
\tilde{a}_{t}^{\ast } \\ 
Y_{t}
\end{array}
\right] ,du_{t}=\left[ 
\begin{array}{c}
dB\left( t\right)  \\ 
dB\left( t\right) ^{\ast }
\end{array}
\right] ,\qquad \mathsf{A}=\left[ 
\begin{array}{ccc}
-\left( \frac{1}{2}\gamma +i\omega \right)  & 0 & -\frac{1}{2}\sqrt{\gamma }%
\lambda  \\ 
0 & -\left( \frac{1}{2}\gamma -i\omega \right)  & -\frac{1}{2}\sqrt{\gamma }%
\lambda  \\ 
\sqrt{\gamma } & \sqrt{\gamma } & 2\lambda 
\end{array}
\right] ,\mathsf{B}=\left[ 
\begin{array}{cc}
-\sqrt{\gamma }0 & 0 \\ 
0 & -\sqrt{\gamma } \\ 
e^{i\theta } & e^{-i\theta }
\end{array}
\right] .
\end{eqnarray*}
The solution is then $x_{t}=e^{\mathsf{A}t}x_{0}+\int_{0}^{t}e^{\mathsf{A}%
\left( t-s\right) }\mathsf{B}du_{s}$. The matrix $\mathsf{A}$ has determinant 
$2\omega ^{2}\lambda $ which is non-degenerate for $\lambda \neq 0$ provided 
$\omega \neq 0$.

We will concentrate on the case $\omega =0$, where the eigenvalues of $%
\mathsf{A}$ are readily calculated to be $0,-\frac{1}{2}\gamma ,-\frac{1}{2}%
\gamma +2\lambda $. We see that $\mathsf{A}$ is (marginally) stable provided
that $\lambda <\frac{1}{4}\gamma $. The solution is
\begin{eqnarray*}
\tilde{a}_{t}=f\left( t\right) a+g\left( t\right) a^{\ast
}+\int_{0}^{t}\left( e^{i\theta }k\left( t-s\right) -\sqrt{\gamma }f\left(
t-s\right) \right) dB\left( t\right) +\int_{0}^{t}\left( e^{-i\theta
}k\left( t-s\right) -\sqrt{\gamma }g\left( t-s\right) \right) dB\left(
t\right) ^{\ast }
\end{eqnarray*}
with
\begin{eqnarray*}
f\left( t\right)  &=&-\frac{2\lambda }{\gamma -4\lambda }+\frac{\gamma }{%
\gamma -2\lambda }e^{-(\frac{1}{2}\gamma -2\lambda )t}+\frac{1}{2}e^{-\frac{1%
}{2}\gamma t}, \\
g\left( t\right)  &=&-\frac{2\lambda }{\gamma -4\lambda }+\frac{\gamma }{%
\gamma -2\lambda }e^{-(\frac{1}{2}\gamma -2\lambda )t}-\frac{1}{2}e^{-\frac{1%
}{2}\gamma t}, \\
k\left( t\right)  &=&\frac{\sqrt{\gamma }\lambda }{\gamma -4\lambda }\left(
e^{-(\frac{1}{2}\gamma -2\lambda )t}-1\right) .
\end{eqnarray*}
We also find that
\begin{eqnarray*}
Y_{t}=\int_{0}^{t}\left( e^{i\theta }r\left( t-s\right) +\gamma p\left(
t-s\right) \right) dB\left( s\right) +\int_{0}^{t}\left( e^{-i\theta
}r\left( t-s\right) +\gamma p\left( t-s\right) \right) dB\left( s\right)
^{\ast }
\end{eqnarray*}
with
\begin{eqnarray*}
p\left( t\right) =-\frac{2}{\gamma -4\lambda }\left( 1-e^{-(\frac{1}{2}%
\gamma -2\lambda )t}\right) ,\quad r\left( t\right) =-\frac{1}{\gamma
-4\lambda }\left( \gamma -4\lambda e^{-(\frac{1}{2}\gamma -2\lambda
)t}\right) .
\end{eqnarray*}

\section{Photon Number Feedback}
An alternative choice is to measure the photon number of the output field, $\Lambda^{ \text{out}}(t)$.
(For convenience, we will treat the $n=1$ input field case.)
This means that we set $Z_t \equiv \Lambda (t)$. In the case where the flow is determined by fixed $SLH$-components
on the initial algebra $\mathfrak{A}_0$, we have 
\begin{eqnarray}
d\Lambda^{\text{out}} (t) = d\Lambda (t) + j_t (S^\ast L) dB^\ast (t) +j_t (L^\ast S)
dB(t) + j_t (L^\ast L) dt .
\label{eq:Lambda_uncontrolled}
\end{eqnarray}

We now show that we obtain the same sort of consistency result we had for quadrature measurements from Proposition \ref{prop:io_quad}. 

\begin{proposition}
Using the same notation as in Proposition \ref{prop:io_quad}, and taking $Z \equiv \Lambda$, we have that the output number operator
for the controlled dynamics is
\begin{eqnarray}
d\Lambda^{\text{out}\left[ Y\right] } (t) \equiv 
d\Lambda (t) + \tilde{S}_t^\ast  \tilde{L}_t \, dB^\ast (t) 
+ \tilde{L}_t^\ast \tilde{S}_t \, dB(t) + \tilde{L}_t^\ast \tilde{L}_t dt .
\label{eq:Lambda_Y}
\end{eqnarray}
\end{proposition}
\begin{proof}
A simple application of the quantum Ito\={o} calculus shows that, for a controlled flow, the analogue of (\ref{eq:io_Y}) is
\begin{eqnarray*}
d\Lambda^{\text{out}\left[ Y\right] } (t) &&= 
d\Lambda (t) + \tilde{S}_t^\ast  \tilde{L}_t \, dB^\ast (t) + \tilde{L}_t^\ast \tilde{S}_t \, dB(t) + \tilde{L}_t^\ast \tilde{L}_t dt
\nonumber \\
&&+\Bigg\{\frac{1}{2}\tilde{L}_{t}[[Y]]^{\ast }\bigg[\tilde{\Lambda}^{\text{out}\left[ Y\right] }(t),\tilde{L}_{t}[[Y]]\bigg]
+\frac{1}{2}\sum_{i}\bigg[\tilde{L}_{t}[[Y]]^{\ast },\tilde{\Lambda}^{\text{out}\left[ Y\right] } (t )\bigg]\tilde{L}_{t}[[Y]]
-i\bigg[\tilde{\Lambda}^{\text{out}\left[ Y\right] }(t),\tilde{H}_{t}[[Y]]\bigg]\Bigg\}dt
\nonumber
\\
&&+\tilde{S}_{t}[[Y]]^{\ast }\bigg[\tilde{\Lambda}^{\text{out}\left[ Y\right] } (t) ,\tilde{L}_{t}[[Y]]\bigg]\,dB^{\ast }\left( t\right) 
\nonumber \\
&&+\bigg[\tilde{L}_{t}[[Y]]^{\ast },\tilde{\Lambda}^{\text{out}\left[ Y\right] } (t) \bigg]\tilde{S}_{t}[[Y]]\,dB \left( t\right)   \nonumber
\\
&&+\bigg(\tilde{S}_{t}[[Y]]^{\ast }\tilde{\Lambda}^{\text{out}\left[ Y\right] }(t)\tilde{S}_{t}[[Y]]-\tilde{\Lambda}^{\text{out}\left[ Y\right] }(t) \bigg)\,d\Lambda  \left( t\right) .
\end{eqnarray*}
Fortunately, the new terms vanish for a fairly simple reason. If we take one of the terms, say $\bigg[\tilde{L}_{t}[[Y]]^{\ast },\tilde{\Lambda}^{\text{out}\left[ Y\right] } (t) \bigg]$, then we note that this corresponds to
\begin{eqnarray*}
\bigg[\tilde{L}_{t}[[Y]]^{\ast },\tilde{\Lambda}^{\text{out}\left[ Y\right] } (t) \bigg]=
U^{[Z] \ast}_t \bigg[ L_{t}[[Z]]^{\ast }, \Lambda  (t) \bigg] U^{[Z]}_t ,
\end{eqnarray*}
but the almost trivial observation at this point is that we have taken $Z \equiv \Lambda$ the present case,
and so $\Lambda (t)$ will commute with $S_{t}[[Z]] , L_{t}[[Z]],  H_{t}[[Z]]$
and their adjoints. This leaves us with (\ref{eq:Lambda_Y}) as claimed.
\end{proof}

We see that (\ref{eq:Lambda_Y}) is structurally identical to (\ref{eq:Lambda_uncontrolled}).

The Heisenberg equations under the controlled flow however will be identical to those derived for the
quadrature case, but with $Y$ now an inhomogeneous Poisson process rather than a diffusion.

\section{Quantum PID Filter}

In this section, we show how to describe one of the basic control feedback
loop mechanisms, PID controllers, in the quantum domain. In its classical
form, we have a modulated control signal of the form
\begin{eqnarray*}
\dot{W}_{t}=k_{P}Y_{t}+k_{I}\int_{0}^{t}Y_{s}ds+k_{D}\dot{Y}_{t}
\label{eq:W_PID}
\end{eqnarray*}
which is the sum of three terms: one proportional to $Y$, one an integral of 
$Y$, and one the derivative of $Y$. The proportional and integral terms can be
modelled following the theory set out in this paper. The derivative term
however is more singular and has to be treated separately. 

\subsection{Quadrature Feedback}

We consider the case $Y_{t}=B_{t}^{\text{out}}+B_{t}^{\text{out}\ast }$
corresponding to quadrature measurement. Here we must set $Z_{t}=B_{t}+B_{t}^{\ast }$ in the input picture. We shall replace the coefficients $k_{P},k_{I},k_{D}$ now with self-adjoint operators $F_{P},F_{I}
$ and $F_{D}$ respectively in $\frak{A}_{0}$.

Our choice for the adapted $SLH$ coefficients will be
\begin{eqnarray*}
S_{t}\left[ \left[ Z\right] \right]  &=&I\otimes I, \\
L_{t}\left[ \left[ Z\right] \right]  &=&L_{0}\otimes I-iF_{D}\otimes I, \\
H_{t}\left[ \left[ Z\right] \right]  &=&H_{0}\otimes I+\frac{1}{2}\left(
F_{D}L_{0}+L_{0}^{\ast }F_{D}\right) \otimes I +F_{P}\otimes Z_{t}+F_{I}\otimes
\int_{0}^{t}Z_{s}ds.
\end{eqnarray*}
The basis for this is that the derivative term is treated in the same way as
in our description of Wiseman's proportional to $\dot{Y}$ feedback: it
results in an additional term in the $L$-operator, and an addition to the
Hamiltonian. We may view this as a bare model $G_{0}\sim \left(
I,L_{0},H_{0}\right) $ into which we feedback the measurement process.

For the output noise, we have $dB^{\text{out}}\left( t\right) =dB\left(
t\right) +j_{t}^{\left[ \left[ Y\right] \right] }\left( L_{0}-iF_{D}\right)
dt$ so that $dY_{t}=dZ_{t}+j_{t}^{\left[ \left[ Y\right] \right] }\left(
L_{0}+L_{0}^{\ast }\right) dt$: that is the $F_{D}$ terms vanish in
accordance with our consistency results from earlier.

The Heisenberg dynamics is then
\begin{eqnarray*}
dj_{t}^{\left[ \left[ Y\right] \right] }\left( X\right)  &=&j_{t}^{\left[ %
\left[ Y\right] \right] }\left( \mathscr{L}X\right) dt-ij_{t}^{\left[ \left[
Y\right] \right] }\left( \left[ X,F_{P}\right] \right) Y_{t}\, dt -ij_{t}^{\left[ %
\left[ Y\right] \right] }\left( \left[ X,F_{I}\right] \right) \left(
\int_{0}^{t}Y_{s}ds\right) dt \\
&&+j_{t}^{\left[ \left[ Y\right] \right] }\left( \left[ X,L_{0}-iF_{D}\right]
\right) dB_{t}^{\ast }+j_{t}^{\left[ \left[ Y\right] \right] }\left( \left[
L_{0}^{\ast }+iF_{D},X\right] \right) dB_{t}.
\end{eqnarray*}
where
\begin{eqnarray*}
\mathscr{L}X &=&\frac{1}{2}\left( L_{0}^{\ast }+iF_{D}\right) \left[
X,L_{0}-iF_{D}\right] +\frac{1}{2}[L_{0}^{\ast }+iF_{D},X]\left(
L_{0}-iF_{D}\right) -i\left[ X,H_{0}+\frac{1}{2}F_{D}L_{0}+\frac{1}{2}%
L_{0}^{\ast }F_{D}\right]  \\
&\equiv &\mathscr{L}_{0}X-\frac{1}{2}\left[ \left[ X,F_{D}\right] ,F_{D}\right] +i%
\left[ F_{D},X\right] L_{0}-iL_{0}^{\ast }\left[ X,F_{D}\right] .
\end{eqnarray*}
Here $\mathscr{L}_{0}$ is the bare GKS-Lindblad generator determined by
coupling $L_{0}$ and Hamiltonian $H_{0}$. With is we may write
\begin{eqnarray*}
dj_{t}^{\left[ \left[ Y\right] \right] }\left( X\right)  &=&j_{t}^{\left[ %
\left[ Y\right] \right] }\left( \mathscr{L}_{0}X\right) dt-\frac{1}{2}j_{t}^{%
\left[ \left[ Y\right] \right] }\left( \left[ \left[ X,F_{D}\right] ,F_{D}%
\right] \right) dt-ij_{t}^{\left[ \left[ Y\right] \right] }\left( \left[
X,F_{P}\right] \right) Y_{t} \, dt -ij_{t}^{\left[ \left[ Y\right] \right] }\left( %
\left[ X,F_{I}\right] \right) \left( \int_{0}^{t}Y_{s}ds\right) dt \\
&&-i\left( dB_{t}+j_{t}^{\left[ \left[ Y\right] \right] }\left( L_{0}\right)
dt\right) ^{\ast }\,j_{t}^{\left[ \left[ Y\right] \right] }\left( \left[
X,F_{D}\right] \right) -i\,j_{t}^{\left[ \left[ Y\right] \right] }\left( %
\left[ X,F_{D}\right] \right) \,\left( dB_{t}+j_{t}^{\left[ \left[ Y\right] %
\right] }\left( L_{0}\right) dt\right)  \\
&&+j_{t}^{\left[ \left[ Y\right] \right] }\left( \left[ X,L_{0}\right]
\right) dB_{t}^{\ast }+j_{t}^{\left[ \left[ Y\right] \right] }\left( \left[
L_{0}^{\ast },X\right] \right) dB_{t}.
\end{eqnarray*}
Note that 
\begin{eqnarray*}
&&-i\left( dB_{t}+j_{t}^{\left[ \left[ Y\right] \right] }\left( L_{0}\right)
dt\right) ^{\ast }\,j_{t}^{\left[ \left[ Y\right] \right] }\left( \left[
X,F_{D}\right] \right) -i\,j_{t}^{\left[ \left[ Y\right] \right] }\left( %
\left[ X,F_{D}\right] \right) \,\left( dB_{t}+j_{t}^{\left[ \left[ Y\right] %
\right] }\left( L_{0}\right) dt\right)  \nonumber\\
&=&-i\left( dB_{t}^{\text{out}%
}+j_{t}^{\left[ \left[ Y\right] \right] }\left( F_{D}\right) dt\right)
^{\ast }\,j_{t}^{\left[ \left[ Y\right] \right] }\left( \left[ X,F_{D}\right]
\right) -i\,j_{t}^{\left[ \left[ Y\right] \right] }\left( \left[ X,F_{D}%
\right] \right) \,\left( dB_{t}^{\text{out}}+j_{t}^{\left[ \left[ Y\right] %
\right] }\left( F_{D}\right) dt \right)  \\
&=&-i\,j_{t}^{\left[ \left[ Y\right] \right] }\left( \left[ X,F_{D}\right]
\right) \left( dB_{t}+dB_{t}^{\text{out}\ast }\right) +j_{t}^{\left[ \left[ Y%
\right] \right] }\left( \left[ \left[ X,F_{D}\right] ,F_{D}\right] \right) dt
\\
&=&-i\,j_{t}^{\left[ \left[ Y\right] \right] }\left( \left[ X,F_{D}\right]
\right) \,dY_{t}+j_{t}^{\left[ \left[ Y\right] \right] }\left( \left[ \left[
X,F_{D}\right] ,F_{D}\right] \right) dt
\end{eqnarray*}
so that we obtain
\begin{eqnarray}
dj_{t}^{\left[ \left[ Y\right] \right] }\left( X\right)  &=&j_{t}^{\left[ %
\left[ Y\right] \right] }\left( \mathscr{L}_{0}X\right) dt+j_{t}^{\left[ %
\left[ Y\right] \right] }\left( \left[ X,L_{0}\right] \right) dB_{t}^{\ast
}+j_{t}^{\left[ \left[ Y\right] \right] }\left( \left[ L_{0}^{\ast },X\right]
\right) dB_{t} \nonumber \\
&&+\frac{1}{2}j_{t}^{\left[ \left[ Y\right] \right] }\left( \left[ \left[
X,F_{D}\right] ,F_{D}\right] \right) dt-i\left\{ j_{t}^{\left[ \left[ Y%
\right] \right] }\left( \left[ X,F_{P}\right] \right) Y_{t}dt +j_{t}^{\left[ %
\left[ Y\right] \right] }\left( \left[ X,F_{I}\right] \right) \left(
\int_{0}^{t}Y_{s}ds\right) dt+j_{t}^{\left[ \left[ Y\right] \right] }\left( %
\left[ X,F_{D}\right] \right) \,dY_{t}\right\} .\nonumber \\
\end{eqnarray}

We therefore obtain the bare dynamics with the desired PID contribution - the
term in braces. We also pick up a back action term $\frac{1}{2}\left[ \left[
X,F_{D}\right] ,F_{D}\right] $.

\bigskip 

As an example, we can consider a cavity mode $a$ with $L_{0}=\sqrt{\gamma
_{0}}a$ and $H_{0}=\omega _{0}a^{\ast }a$ and $F_{P}=k_{P} (a +a^{\ast })$, $%
F_{I}=k_{I}(a+a^{\ast })$ and $F_{D}=k_{D}(a+a^{\ast })$. Setting $\tilde{a}%
_{t}=j_{t}^{\left[ \left[ Y\right] \right] }\left( a\right) $ gives
\begin{eqnarray}
d\tilde{a}_{t}=-\left( \frac{1}{2}  \gamma _{0} 
+i\omega _{0}\right) \tilde{a}_{t}dt-\sqrt{\gamma _{0}}dB_{t}-i dW_t ,
\end{eqnarray}
with $dW_t = k_{P}Y_{t}dt+k_{I}\left( \int_{0}^{t}Y_{s}ds\right) dt+k_{D}dY_{t}  $ corresponding to
the PID filtered process as in (\ref{eq:W_PID}).

Alternatively, we could take $F_{P}=k_{P}a^{\ast }a$, $F_{I}=k_{I}a^{\ast }a$ and $F_{D}=k_{D}a^{\ast }a$,
then we find
\begin{eqnarray}
d\tilde{a}_{t}=-\left( \frac{1}{2}\left( \gamma _{0}-k_{D}^{2}\right)
+i\omega _{0}\right) \tilde{a}_{t}dt-\sqrt{\gamma _{0}}dB_{t}-i 
\tilde{a}_{t} \, dW_t,
\end{eqnarray}
with $W_t$ again being the PID filtered process in (\ref{eq:W_PID}).
Here the derivative term has altered the damping strength. Otherwise, the
three terms enter just into the Hamiltonian $H_{t}\left[ \left[ Z\right] %
\right] $ leading to a PID filtered version of $Y$ entering as an additional
term to the frequency $\omega _{0}$.

\section{Quantum Feedback Network Rules for Controlled Models}
In this section we make some rudimentary observations about what the quantum feedback network rules
should look like when the various $SLH$-coefficients of the components are allowed to be adapted processes.
We stress that we can only give a sketch of mathematics behind this - the problem of working with a fully
rigorous model establishing the self-adjointness of the underlying Hamiltonian and the associated
instantaneous feedback limits is well beyond current mathematics in quantum probability. However,
leaving aside any pretense at rigor, we can make some reasonable deductions on what to expect.

The derivation of the quantum feedback network rules relies heavily on the
Hamiltonian formulation of quantum stochastic calculus derived by Chebotarev 
\cite{Cheb97}, for the case of commuting coupling coefficients, and by
Gregoratti \cite{Gregoratti} for the general bounded operator case. (The
requirement of boundedness was later lifted \cite{QG}.) The unitary
stochastic process $U_{t}$ generated by coefficients $G\sim \left(
S,L,H\right) $ was shown to be a singular perturbation of the generator of
the free shift along the $x$-axis (with the input being the positive axis
and the output line being the negative). 
We have the space $\mathfrak{h}_{0}\otimes \Gamma \left( L_{%
\mathbb{C}^{n}}^{2}\left( \mathbb{R}\right) \right) $ which may be thought
of as consisting of vectors $\Psi =\left( \Psi _{m}\right) _{m\geq 0}$ where 
$\Psi _{m}$ is in $\mathfrak{h}_0\otimes \left( \bigotimes_{\mathrm{symm.}%
}^{m}L_{\mathbb{C}^{n}}^{2}(\mathbb{R})\right) $. That is, we have the
functions 
\begin{eqnarray*}
\Psi _{m} &:&\mathbb{R}^{m}\times \{1,\cdots ,n\}\mapsto \mathfrak{h}_0, \\
&:&(t_{1},\cdots ,t_{m},k_{1},\cdots ,k_{m})\mapsto \Psi _{k_{1},\cdots
,k_{m}}(t_{1},\cdots ,t_{m}),
\end{eqnarray*}
which are symmetric under interchange of the parameters $(t_{j},k_{j})$. We
define the operators 
\begin{eqnarray*}
\left( b_{k}(s)\,\Psi \right) _{k_{1},\cdots ,k_{m}}(t_{1},\cdots
,t_{m})=\Psi _{k_{1},\cdots ,k_{m},k_{m+1}=k}(t_{1},\cdots ,t_{m},t_{m+1}=s),
\end{eqnarray*}
where appropriate along with the one-sided annihilators $b_{i}(0^{\pm })$.
Specifically, their Hamiltonian, $K$
is given by 
\begin{eqnarray}
-iK\Psi =-i\tilde{K}_{0}\Psi -(\frac{1}{2}\sum_{k}L_{k}^{\ast }L_{k}+iH)\Psi
-\sum_{kj}L_{k}^{\ast }S_{kj}b_{j}(0^{+})\Psi ,
\end{eqnarray}
where 
\begin{eqnarray*}
\tilde{K}_{0}=\left( \int_{-\infty }^{0}+\int_{0}^{+\infty }\right)
b_{j}(x)^{\ast }\left( i\frac{\partial }{\partial x}\right) b_{j}(x)\,dx,
\end{eqnarray*}
Note that $\tilde{K}_{0}$ generates of translation down the $x$-axis and is
a second quantisation of the momentum operator. The states $\Psi $ live in
the Hilbert space $\mathfrak{h}_{0}\otimes \Gamma \left( L_{\mathbb{C}%
^{n}}^{2}\left( \mathbb{R}\right) \right) $ and the domain of $K$ consists
of suitably regular vectors which satisfy a supplementary boundary condition 
\begin{eqnarray}
b_{j}(0^{-})\Psi =L_{j}\Psi +\sum_{k}S_{jk}\,b_{k}(0^{+})\Psi .
\label{eq:BC}
\end{eqnarray}

The quantum feedback network theory builds on this to consider a graph with
separate quantum systems at the vertices (as separate $G_{k}\sim \left(
S_{k},L_{k},H_{k}\right) $ for each one) and propagating Bose fields along
the edges. On each edge we have a sense of direction of propagation. Some of
the edges run between vertices - these are the internal ones - while some
extend to infinity. The semi-infinite edges then correspond to either input
lines (terminating at a vertex) or an output lines (starting at a vertex).
The lines can have arbitrary multiplicity for the number of Bose fields they
carry and the vertices can have several incoming and outgoing edges; but the
total multiplicity in must equal the total multiplicity out. Feedback
reduction then takes place by shortening each of the edges down to zero
length (instantaneous feedback limit). At each vertex, we have a boundary
condition of the form (\ref{eq:BC}). When we eliminate an edge, we get a
reduction of order of the graph and a telescoping of the boundary conditions
in a systematic manner. Eliminating all internal edges should result in a
Markovian model which is an effective model for the network in the
instantaneous propagation limit.

For the case of a controlled flow for a single component $G_{t}\left[ \left[
Z\right] \right] \sim \left( S_{t}\left[ \left[ Z\right] \right] ,L_{t}\left[
\left[ Z\right] \right] ,H_{t}\left[ \left[ Z\right] \right] \right) $
considered in this paper it is natural to argue that the corresponding
Hamiltonian should be
\begin{eqnarray*}
-iK_{t}\left[ \left[ Z\right] \right] \,\Psi =-i\tilde{K}_{0}\Psi -(\frac{1}{%
2}\sum_{k}L_{k,t}^{\ast }\left[ \left[ Z\right] \right] L_{k,t}\left[ \left[
Z\right] \right] +iH_{t}\left[ \left[ Z\right] \right] )\Psi
-\sum_{kj}L_{k,t}^{\ast }\left[ \left[ Z\right] \right] S_{kj,t}\left[ \left[
Z\right] \right] b_{j}(0^{+})\Psi ,
\end{eqnarray*}
with the boundary condition
\begin{eqnarray*}
b_{j}(0^{-})\Psi =L_{j,t}\left[ \left[ Z\right] \right] \Psi
+\sum_{k}S_{jk,t}\left[ \left[ Z\right] \right] \,b_{k}(0^{+})\Psi .
\end{eqnarray*}
Note that the Hamiltonian, and its boundary condition are time dependent.
They also depend on the control process $Z$ however this is, in principle,
amenable to rigorous formulation. Note that it is essential that we
construct the Hamiltonian $K$ in the input picture!

The step up to a non-Markovian network of such is now evident, and in
principle the edge elimination should proceed in a similar manner as before
leading to the same algebraic form for the rules.

\bigskip 

For instance, the series product for systems $G_{t}^{\left( A\right) }\left[ %
\left[ Z\right] \right] $ and $G_{t}^{\left( B\right) }\left[ \left[ Z\right]
\right] $ should be
\begin{eqnarray*}
&& \left( S_{t}^{B} \left[ \left[ Z\right] \right] ,L_{t}^{B}\left[ \left[ Z\right] %
\right] ,H_{t}^{B}\left[ \left[ Z\right] \right] \right) \vartriangleleft \left(
S_{t}^{A}\left[ \left[ Z\right] \right] ,L_{t}^{A}\left[ \left[ Z\right] \right]
,H_{t}^{A}\left[ \left[ Z\right] \right] \right) \nonumber \\
 &=& \bigg( S_{t}^{B }%
\left[ \left[ Z\right] \right] S_{t}^{A }\left[ \left[ Z\right]
\right] ,L_{t}^{B }\left[ \left[ Z\right] \right]
+S_{t}^{B }\left[ \left[ Z\right] \right] L_{t}^{A }\left[ \left[ Z\right] \right] ,  
H_{t}^{A }\left[ \left[ Z\right] \right] +H_{t}^{B }\left[ \left[ Z\right] \right] +\text{Im} L_{t}^{B}%
\left[ \left[ Z\right] \right] S_{t}^{B }\left[ \left[ Z\right]
\right] L_{t}^{A }\left[ \left[ Z\right] \right] \bigg) .
\end{eqnarray*}

\begin{figure}[htbp]
	\centering
		\includegraphics[width=0.5\textwidth]{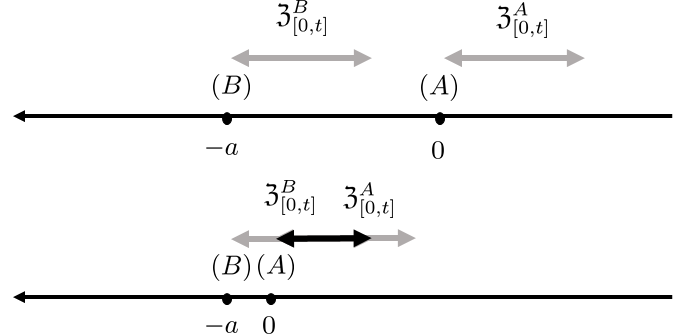}
	\caption{Systems in series separated by distance $a$. The control process algebras are sketched for $0<t<a$ and $t>a$.}
	\label{fig:NM_QFN_II_Series_Limit}
\end{figure}

The cascaded systems may be considered as separated by a distance $a$ before the limit. Taking the speed of propagation to
be $c=1$, we see that the input algebras to $A$ and $B$ at time $t$ are algebras generated by $\{ Z_s: 0\le s \le t\}$
and $\{ Z_s : -a \le s \le t-a\}$ respectively. For $t$ less that $a$, there is no intersection, however, for $t>a$
they overlap and after $a \to 0^+$, for $t$ fixed, the two algebras must coincide, see Figure \ref{fig:NM_QFN_II_Series_Limit}.

We can also consider a beam splitter, as in Figure \ref{fig:NM_QFN_II_beam_splitter}. While there are no interconnections
here, it is worth mentioning that we may want to select different measurements at he two output ports. For instance, at the first output we may measure the quadrature $Y_1(t) = B^{\text{out}}_1(t)+ B^{\text{out} \ast}_1 (t)$ and at the second perform a
photon counting measurement $Y_2 (t) = \Lambda_{22}(t)$. In this case we should take $Z_1 (t) = B_t (t) + B_i^\ast (t)$
and $Z_2 (t) = \Lambda_{22} (t)$. This may seem strange as we are not measuring the inputs, and the output fields are a linear
superposition of the input fields, however, this is the construction we have to make if we revert to the input picture. In this
case,the two measurement algebras are independent factors of the total noise algebra.

\begin{figure}[htbp]
	\centering
		\includegraphics[width=0.25\textwidth]{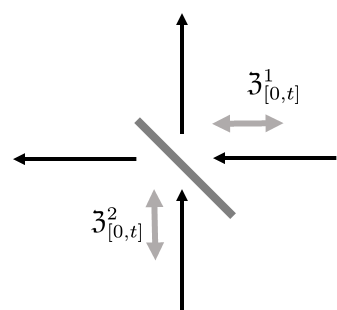}
	\caption{Inputs into a beam splitter.}
	\label{fig:NM_QFN_II_beam_splitter}
\end{figure}

In Figure \ref{fig:NM_QFN_II_in_loop} we have a simple network consisting of a beam splitter $B$ and a cavity $A$.
The cavity is put into an optical feedback loop using the beam splitter. While the time to propagate around the 
loop is finite, the model will be non-Markovian - unless the loop field is somehow incorporated into the system.
We have sketched the sections of the edges where the input algebras for $A$ and $B$ live, and for $t$ less that the propagation
time from $B$ to $A$ they do not intersect. In the instantaneous feedback limit the loop shrinks to zero.
As such, the algebra $\mathfrak{Z}^A_{[0,t]}$ gets pushed back out of the loop and eventually coincides with 
$\mathfrak{Z}^B_{[0,t]}$.

\begin{figure}[htbp]
	\centering
		\includegraphics[width=0.25\textwidth]{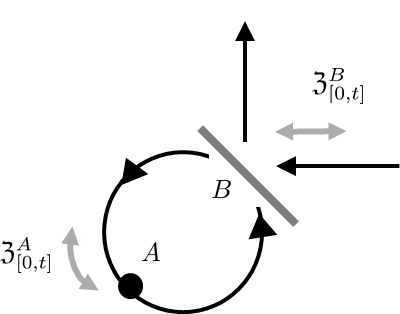}
	\caption{Algebraic feedback loop}
	\label{fig:NM_QFN_II_in_loop}
\end{figure}

One therefore reasonably expects the new quantum feedback network rules to be formally the same as for the 
standard concatenation and feedback reduction rules appearing in \cite{GouJam09a}, with the modeling proviso
that functional form the components in terms of $Y$ is decided upon in the output picture, then deduced for the input picture and 
incorporated in to the $SLH$-coefficients to make them dependent on the various control processes $Z$.

\section{Conclusion}
Our discussions have shown that there is an inner consistency within the quantum feedback set-up provided
one correctly distinguishes between the input picture and the output picture. Both are invaluable as far
as model building and analysis are concerned, but it is necessary to understand the connection between these
to have an overview of quantum feedback systems.

Unlike the Wiseman paper, where the feedback Hamiltonian is proportional to $\dot{Y}$ (in the output picture),
we consider feedback models that are regular. Wiseman's theory has effectively been re-derived in the input picture 
by us in our paper introducing the series product construction \cite{GouJam09b}. The present paper opens up
a more general theory of quantum feedback where a modulating filter processes the measurement readout, and
allows us to consider a broader range of feedback scenarios. 

A noteworthy feature is that we have not had to go into any specificity about the modulating filter. It could be
classical, in which case we just need the input-output relation giving $W_t$ as process adapted to $Y$. And
in particular, we do not have to worry about the mechanism by which the filter works, or how we would model
a hybrid classical and quantum system. It could also be quantum in which case its degrees of freedom would
have to be taken into account - however these would be independent of the system's and would easily be 
handled by an augmentation of the results presented in \ref{sec:ConDyn} absorbing the algebra of the
filter observables which of course commutes with $\mathfrak{A}_0$ and $\mathfrak{N}_{[0,\infty )}$.

The theory set out here provides for great flexibility in modeling and designing quantum control systems.
The issue of how exactly we physically realize the filter, or how exactly we couple the modulated signal to
the $SLH$-coefficients (that is, engineer a specific $Z$-dependence on the $SLH$-coefficients) has not been 
addressed here. However, the power of the theory is that the specifics are not particularly relevant, and that
a consistent systematic approach exists which is not dependent on these details.  This now enables us to
apply a wide range of standard control engineering methods to quantum open systems.

\bigskip

\textbf{Acknowledgements:} The author gratefully acknowledges several key discussions relating to this paper:
especially with Luc Bouten on the quantum feedback network rules for adapted $SLH$ coefficients as far back as 2011
and up to recent times,
with Hendra Nurdin, and with Matt James on the need for a theory of quantum feedback flexible enough for a general
formulation of quantum control engineering applications.


\begin{thebibliography}{99}
\bibitem{argument} This of course depends on whether you take $Y$ or $y=\dot Y$ as the error signal.
Mathematically, one might argue that $y = \dot Y$ is already singular enough that you would not want to differentiate any further,
so that a proportional-integral-derivative (PID) controller should be $Y(t), \int_0^t Y(s)ds, \dot{Y} (t)$ respectively.
\end{thebibliography}
\end{document}